\documentclass[a4paper,11pt]{article}
\usepackage[bbgreekl]{mathbbol}
\usepackage[small]{titlesec}
\usepackage[utf8]{inputenc}
\usepackage[T1]{fontenc}  
\usepackage[backend=biber,style=numeric,natbib=false,giveninits=true,doi=true,isbn=false,url=false,date=year,maxbibnames=99]{biblatex}
\DeclareNameAlias{default}{family-given}

\DeclareFieldFormat[article]{title}{#1}
\renewbibmacro{in:}{}
\DeclareFieldFormat[article]{pages}{#1}
\DeclareFieldFormat{journal}{#1}
\renewcommand\bf\bfseries

\renewbibmacro*{volume+number+eid}{%
	\printfield{volume}%
	\setunit*{\addnbspace}% NEW (optional); there's also \addnbthinspace
	\printfield{number}%
	\setunit{\addcomma\space}%
	\printfield{eid}}
\DeclareFieldFormat[article]{number}{\mkbibparens{#1}}
\DeclareFieldFormat[article]{volume}{\textbf{#1}}
\DeclareFieldFormat{year}{\mkbibparens{#1}}
\DeclareBibliographyDriver{article}{%
	\printnames{author}:%
	\newunit\newblock
	\printfield{title}%
	\newunit\newblock
	\printfield{journaltitle}
	\newunit
	\iffieldundef{number}{\printfield{volume}}{\printfield{volume}\addspace\printfield{number}}
	\addcomma\addspace\printfield{pages}\addspace
	\printfield{year}
}
\AtEveryBibitem{\clearfield{month}}
\AtEveryBibitem{\clearfield{day}}
\usepackage[english]{babel}
\usepackage[T1]{fontenc}
\usepackage{csquotes}
\usepackage{bbm}
\usepackage[leqno]{amsmath}
\usepackage{amsfonts,amsthm,amsbsy,amssymb,dsfont,stmaryrd}
\usepackage[dvipsnames]{xcolor}
\usepackage{braket}

\makeatletter
\newcommand{\leqnomode}{\tagsleft@true\let\veqno\@@leqno}
\newcommand{\reqnomode}{\tagsleft@false\let\veqno\@@eqno}
\makeatother

\usepackage{mathtools}

\numberwithin{equation}{section}

\newcommand\myshade{85}
\colorlet{mylinkcolor}{violet}
\colorlet{mycitecolor}{YellowOrange}
\colorlet{myurlcolor}{Aquamarine}

\usepackage[left=2.5cm,right=2.5cm,top=2.5cm,bottom=2.5cm]{geometry}
\usepackage[unicode=true,pdfusetitle,
bookmarks=true,bookmarksnumbered=false,bookmarksopen=false,
breaklinks=false,pdfborder={0 0 1},backref=false,
linkcolor  = mylinkcolor!\myshade!black,
citecolor  = mycitecolor!\myshade!black,
urlcolor   = myurlcolor!\myshade!black,
colorlinks = true,
]
{hyperref}
\usepackage[nameinlink]{cleveref}

\usepackage{tikz}
\usetikzlibrary{arrows}
\usetikzlibrary{intersections}

\usepackage{graphicx}
\usepackage{caption}

\definecolor{ct_black}{HTML}{000000}
\definecolor{ct_orange}{HTML}{ED872D}
\definecolor{ct_purple}{HTML}{7A68A6}
\definecolor{ct_blue}{HTML}{348ABD}
\definecolor{ct_turquoise}{HTML}{188487}
\definecolor{ct_red}{HTML}{E32636}
\definecolor{ct_pink}{HTML}{CF4457}
\definecolor{ct_green}{HTML}{467821}

\definecolor{ct2_green}{HTML}{9FF781}
\definecolor{ct2_green_dark}{HTML}{088A08}

\theoremstyle{plain}
\newtheorem{thm}{\protect\theoremname}[section]
\theoremstyle{plain}
\newtheorem{lem}[thm]{\protect\lemmaname}
\theoremstyle{plain}
\newtheorem{cor}[thm]{\protect\corollaryname}
\theoremstyle{plain}
\newtheorem{prop}[thm]{\protect\propositionname}
\theoremstyle{plain}

\theoremstyle{remark}
\newtheorem{rem}[thm]{\protect\remarkname}
\theoremstyle{definition}
\newtheorem{defn}[thm]{\protect\definitionname}
\theoremstyle{plain}

\providecommand{\assumptionname}{Assumption}
\providecommand{\claimname}{Claim}
\providecommand{\corollaryname}{Corollary}
\providecommand{\definitionname}{Definition}
\providecommand{\lemmaname}{Lemma}
\providecommand{\propositionname}{Proposition}
\providecommand{\remarkname}{Remark}
\providecommand{\theoremname}{Theorem}
\providecommand{\examplename}{Example}

\crefname{section}{Section}{Sections}
\crefname{appendix}{Appendix}{Appendices}
\crefname{figure}{Figure}{Figures}
\crefname{assumption}{Assumption}{Assumptions}
\crefname{thm}{Theorem}{Theorems}
\crefname{lem}{Lemma}{Lemmas}
\crefformat{equation}{(#2#1#3)}

\crefrangelabelformat{equation}{(#3#1#4--#5#2#6)}

\crefmultiformat{equation}{(#2#1#3}{, #2#1#3)}{#2#1#3}{#2#1#3}
\Crefmultiformat{equation}{(#2#1#3}{, #2#1#3)}{#2#1#3}{#2#1#3}

\newtheorem*{lem*}{\protect\lemmaname}

\newcommand{\ee}{\operatorname{e}}
\newcommand{\ii}{\operatorname{i}}
\newcommand{\Mat}{\operatorname{Mat}}
\newcommand{\Herm}{\operatorname{Herm}}
\newcommand{\ZZ}{\mathbb{Z}}

\newcommand{\NN}{\mathbb{N}}
\newcommand{\RR}{\mathbb{R}}
\newcommand{\CC}{\mathbb{C}}
\newcommand{\FF}{\mathbb{F}}
\newcommand{\PP}{\mathbb{P}}
\newcommand{\EE}{\mathbb{E}}
\newcommand{\VV}{\mathbb{Var}}

\newcommand{\calG}{\mathcal{G}}
\newcommand{\calO}{\mathcal{O}}
\newcommand{\calR}{\mathcal{R}}
\newcommand{\calD}{\mathcal{D}}
\newcommand{\calV}{\mathcal{V}}

\newcommand\norm[1]{\left\lVert#1\right\rVert}
\newcommand{\ip}[2]{\langle #1, #2 \rangle}
\newcommand{\HSip}[2]{\langle #1, #2 \rangle_{\mathrm{HS}}}

\newcommand{\dif}{\operatorname{d}}
\newcommand{\tr}{\operatorname{tr}}
\renewcommand{\Im}[1]{\operatorname{\mathbb{I}\mathbbm{m}}\{#1\}}
\renewcommand{\Re}[1]{\operatorname{\mathbb{R}\mathbbm{e}}\{#1\}}
\newcommand{\ve}{\varepsilon}
\newcommand{\vf}{\varphi}
\newcommand{\Id}{\mathds{1}}
\newcommand{\supp}{\operatorname{supp}}

\usepackage{environ}

\NewEnviron{malign}{%
	\begin{align}\begin{split}
			\BODY
	\end{split}\end{align}
}

\newcommand{\prob}[1]{\PP\left[\Set{#1}\right]}
\newcommand{\ex}[1]{\EE\left[#1\right]}
\newcommand{\HSn}[1]{\norm{#1}_{\mathrm{HS}}}

\newcommand{\eq}[1]{\begin{align*}#1\end{align*}}
\newcommand{\eql}[1]{\begin{align}#1\end{align}}

\title{Dynamical Localization for Random Band Matrices\\ up to $W\ll N^{1/4}$}
\author{\href{mailto:gc4233@princeton.edu}{Giorgio Cipolloni}\\
	{\footnotesize Princeton Center for Theoretical Science, Princeton University }\\
	\href{mailto:peledron@tauex.tau.ac.il}{Ron Peled}\\
	{\footnotesize School of Mathematical Sciences, Tel Aviv University}\\
	\href{mailto:schenke6@msu.edu}{Jeffrey Schenker}\\
	{\footnotesize Department of Mathematics, Michigan State University}\\
	\href{mailto:jacobshapiro@princeton.edu}{Jacob Shapiro}\\
	{\footnotesize Department of Physics, Princeton University}
}

\begin{document}
	\reqnomode
	
	\maketitle
	\begin{abstract}
		
		%	\corr [Add emails, addresses, keywords, AMS classes, eventual grants,...] \nc
	%	\corr according to the generalizations we finally include, decide if the phrasing ``large class'' and the omission of the word ``Gaussian'' are still appropriate to describe our results\nc
		
		We prove that a large class of $N\times N$ Gaussian random band matrices with band width $W$ exhibits dynamical Anderson localization at all energies when $W \ll N^{1/4}$. The proof uses the fractional moment method \cite{Aizenman_Molchanov_1993_cmp/1104253939} and an adaptive Mermin--Wagner style shift.
		%We consider a large class of $N\times N$ random band matrices with band width $W$, and prove that for $W \ll N^{1/4}$ they exhibit Anderson localization at all energies. To prove this result, we rely on the fractional moment method \cite{Aizenman_Molchanov_1993_cmp/1104253939}, and on the so-called Mermin-Wagner shift (a common tool in statistical mechanics).
	\end{abstract}
	
	\section{Introduction}
	Let $W\in\NN$ and $\Set{T_j}_{j=1}^{\infty},\Set{V_j}_{j=1}^{\infty}$ be two sets of indepedent and identically distributed $W\times W$ complex matrices distributed according to the Ginibre ensemble and Gaussian Unitary Ensemble ($\mathrm{GUE}$), respectively.  For any $n\in\NN$ define the random $nW\times nW$ matrix 
	\eql{
		\label{eq:Hamiltonian}
		H & :=  \begin{bmatrix}V_{1} & -T_{1}^{\ast} & 0\\
			-T_{1} & V_{2} & -T_{2}^{\ast}\\
			0 & -T_{2} & V_{3}\\
			&  &  & \dots\\
			&  &  &  & \dots\\
			&  &  &  &  & V_{n-2} & -T_{n-2}^{\ast} & 0\\
			&  &  &  &  & -T_{n-2} & V_{n-1} & -T_{n-1}^{\ast}\\
			&  &  &  &  & 0 & -T_{n-1} & V_{n}
		\end{bmatrix}.
	}
	In the present paper, we shall prove
	\begin{thm}\label{thm:first loc theorem} Let $K\subseteq\CC$ be compact. Then there exists an $s_0\in(0,1)$ such that for all $s\in(0,s_0)$ and $z\in K$ there exist $C<\infty,\mu>0$, independent of $n$ and $W$, such that
		\eq{
    	\sup_{z\in K}\ex{\norm{\left(H-z\Id\right)^{-1}_{\quad i,j}}^s} \leq W^C \exp\left(-\mu \frac{|i-j|}{W^4}\right)\qquad(i,j\in\Set{1,\dots,nW})\,.
		}
	\end{thm}
	
	The operator $H$ defined in \cref{eq:Hamiltonian} is the so-called \emph{Wegner W-orbital model}, introduced in \cite{wegnerOrbital1979}, which is most natural for our methods.  \cref{thm:first loc theorem} states that this model exhibits \emph{Anderson localization}, with  localization length $\lesssim W^4$. For this model, standard arguments imply the full range of dynamical localization results; see below for a definition of these notions and \cref{sec:mainres} for the precise formulation of the model and results.  In \cref{sec:Generalizations} we explain how the proof may be modified to deal with the proper random band matrix case (in which case the $T_j$'s are lower triangular), real-valued matrices and other possible generalizations.
	
	\subsection{Background}
	%\paragraph{Background} 
	Anderson localization is a physical phenomenon wherein a quantum particle becomes trapped because of disorder in its potential-energy ``landscape.'' This phenomenon was first studied by Anderson \cite{Anderson_1958_PhysRev.109.1492} using the following model Hamiltonian on $\ell^2(\ZZ^d)$:
	%Anderson localization is a basic phenomenon in solid state physics where an electron becomes trapped in a bounded region in a solid as a consequence of the disorder in the structure of the material, which in turn creates a destructive interference pattern due to the phases of its quantum mechanical wave function \cite{Anderson_1958_PhysRev.109.1492}. To study this, Anderson proposed the following Hamiltonian on $\ell^2(\ZZ^d)$: 
	\eql{ 
		\label{eq:Anderson's model}
		(H\psi)_x := -(T\psi)_x + \lambda V_x\psi_x\qquad(\psi\in\ell^2(\ZZ^d),x\in\ZZ^d) 
	} 
	where $(T\psi)_x := \sum_{\norm{e}=1}\psi_{x+e}$ is the discrete Laplacian (up to a constant) and $\Set{V_x}_{x\in\ZZ^d}$ is a set of independent and identically distributed random variables; the constant $\lambda>0$ parameterizes the strength of the disorder. For $\lambda \gg 1$, Anderson argued that interference effects cause all eigenfunctions of this operator to decay rapidly at infinity, leading to an absence of diffusion, i.e., localization. 
	
	The first mathematical proof of Anderson localization is due to Goldsheid, Molchanov and Pastur \cite{goldsheid1977}, in the context of a related one-dimensional model. For Anderson's model \cref{eq:Anderson's model} with $d=1$, the first proof was due to Kunz and Souillard \cite{Kunz_Souillard_1980_cmp/1103908590}. Then Fr\"ohlich and Spencer \cite{Frohlich_Spencer_1983_cmp/1103922279} used multi-scale analysis to prove that the Green's function 
	\eql{
		\label{eq:Green'sfunction}
		G(x,y;E+\ii\ve)\ = \  \ip{\delta_x}{(H-(E+i\epsilon)\Id)^{-1}\delta_y}
	}
	of Anderson's model (in any dimension $d$) decays exponentially as $\norm{x-y}\to \infty$ with probability one, either at large disorder $\lambda \gg1$ or at extreme energies $|E|\gg1$, with bounds uniform in $\ve>0$. A different proof was obtained by Aizenman and Molchanov \cite{Aizenman_Molchanov_1993_cmp/1104253939,aizenmanLocalizationWeakDisorder1994}, based on showing that \emph{fractional moments} of the Green's function, i.e., $\EE\left[|G(x,y;z)|^s\right]$ with $0<s<1$, decay exponentially.

	The relation between Green's function decay, eigenfunction decay, and dynamical bounds has been explored by a number of authors. Already in \cite{Frohlich_Spencer_1983_cmp/1103922279}, Fr\"ohlich and Spencer proved the absence of diffusion; in follow up work with Martinelli and Scoppola \cite{frohlichConstructiveProofLocalization1985} they proved exponential localization of the eigenfunctions. From fractional moment bounds,  Aizenman \cite{aizenmanLocalizationWeakDisorder1994} proved \emph{dynamical localization}. More precisely, Aizenman showed that if for some interval $I\subseteq\RR$ there is an $s\in (0,1)$ such that at each $E\in I$ we have \eql{
 		\label{eq:fmc}
 		\sup_{\eta>0}\EE\left[|\ip{\delta_x}{\left(H-(E+\ii\eta)\Id\right)^{-1}\delta_y}|^s\right]\leq C \ee^{-\norm{x-y}/\xi}\qquad(x,y\in\ZZ^d) \, 	} 
 	for some $C,\xi<\infty$, then one has \emph{strong exponential dynamical localization on $I$} in the  sense that: 
	\eql{
		\label{eq:absence of diffusion}
		\EE\left[\sup_{t>0}\left|\ip{\delta_x}{\exp(-\ii t H)\chi_{I}(H)\delta_y}\right|\right]\leq C'\ee^{-\norm{x-y}/\xi'}\qquad(x,y\in\ZZ^d) } 
	with $C'\propto C$ and $\xi'\propto \xi$.  The parameter $\xi$ (or $\xi'$) is  referred to as the \emph{localization length}. 
	From \cref{eq:absence of diffusion} one may conclude, e.g., using the RAGE theorem \cite[Theorem 2.6]{AizenmanWarzel2016}, that the spectrum of $H$ in $I$ is pure point with exponentially decaying eigenfunctions. 
	%In this case it is also possible to show the DC electric conductivity of the material associated with that energy region is zero, i.e., the material is an insulator. 
	
% 	Exponential decay of fractional moments of the Green's functions has by become an independent marker of localization, implying pure point spectrum and strong dynamical localization. We shall adopt it as well: $H$ exhibits localization on an interval $I\subset\RR$ if there exists 
	
	Although \emph{dynamical localization} \cref{eq:absence of diffusion} is physically surprising, mathematically, the problem of proving diffusion for energies in the center of the band when $\lambda \ll 1$, and establishing the so-called metal-insulator transition as $\lambda,E$ are varied, remains one of the biggest open problems in mathematical physics.  Furthermore, no transition is expected in dimensions $d=1,2$, where localization is expected at all energies when $\lambda >0$ \cite{Gang_of_4_1979_PhysRevLett.42.673}; this was proved for $d=1$ in \cite{goldsheid1977,Kunz_Souillard_1980_cmp/1103908590} but remains open in $d=2$. 
	%In \cite{Spencer2012-ud} it is explained that in $d=1$ the localization length $\xi$ scales with $\lambda$ as $ \xi \sim \frac{1}{\lambda^2}\, $ as $\lambda \ra 0$.
	For $d=1$, the operator on a finite interval $[1,n]\cap \ZZ$, there is a transition from extended states/GOE statistics to localized states/Poisson statistics if one takes $\lambda \propto 1/\sqrt{n}$; see, e.g., \cite{kritchevski2012,valko2014,rifkind2018}.
	
% 	{\corr Can we prove localization threshold for $W=1$ and study when it happens as $\lambda$ varies with $n$? -- \texttt{Yes.  This is possible.  There is a paper of Virag and Rifkind that looks at exactly this problem, but I don't know that we need to mention it here. --JS}\nc}
	
	To shed light on the, as yet conjectural, metal-insulator transition for $d\geq3$ in \cref{eq:Anderson's model}, it has been proposed \cite{wegnerOrbital1979,PhysRevLett.64.1851,fyodorov1991scaling} to study localization in the context of \emph{random band matrices}, where instead of varying $\lambda$ one replaces Anderson's model \cref{eq:Anderson's model} by a Hamiltonian with random hoppings up to range $W$. Before turning to the proper model which is associated with this name, let us still think for a moment of infinite-volume systems. Consider a self-adjoint operator $H$ on $\ell^2(\ZZ)$ where $\{H_{xy}:=\ip{\delta_x}{H\delta_y}\}_{x>y}$ are independent and identically distributed random complex variables and $\{H_{xx}\}_{x}$ are i.i.d. real variables, and such that \eql{
		\label{eq:infinite volume RBM model} 
		H_{xy} = 0\qquad(x,y\in\ZZ^d:\norm{x-y}> W) 
	} for some range $W\in\NN$. To keep the spectrum of $H$ of order $1$ (as $W\to\infty$), one may choose the entries to be mean-zero with variance of order $1/W$. In these models, the range of the hopping $W$ replaces the disorder strength $\lambda$. It is conjectured (e.g., \cite{PhysRevLett.64.1851,fyodorov1991scaling}) that $\xi$, the localization length, should depend on the range of the hopping as 
	\eql{
		\label{eq:conjecture for localization length of infinite volume RBM}
		\xi \propto W^2\,.
	} These models may be considered as quasi-one-dimensional, i.e., as being defined on the strip of width $W$, $\ell^2(\ZZ)\otimes\CC^W$, with a Hamiltonian given by 
	\eql{ 
		\label{eq:infinite volume one d RBM model}
		(H\psi)_j := -T_{j}^\ast \psi_{j+1}-T_{j-1}\psi_{j-1}+V_j\psi_j\qquad(\psi:\ZZ\to\CC^W;j\in\ZZ) 
	} 
	where $\Set{T_j}_j$ (resp. $\Set{V_j}_j$) are random $W\times W$ complex triangular (resp. Hermitian) matrices. Such models are known to be completely localized \cite{KLEIN1990135,Carmona1987,shapiro2021incomplete,macera2021anderson} for any $W$, but the localization length in these studies is not estimated quantitatively.
	
	If, however, we restrict to finite volume, by truncating the operator in \cref{eq:infinite volume one d RBM model} to the Hilbert space $\ell^2(\Set{1,\dots,n})\otimes\CC^W\cong\CC^{nW}$ for some $n$, we get the proper one-dimensional \emph{random band matrix} model, which is a matrix of size $N:=nW$ and band width $W$ presented in \cref{eq:Hamiltonian} above.
	%For finite $n$, this model exhibits a transition although it is one-dimensional, unlike its infinite volume counterpart. 
	If both asymptotic parameters $W,N$, with $W\leq N$, go to infinity simultaneously, \cite{PhysRevLett.64.1851, PhysRevE.48.R1613, fyodorov1991scaling, fyodorov1994statistical} the following two distinct behaviors are expected depending on how this limit is taken:
	\begin{enumerate}
		
		\item \label{it:1} For $W\ll \sqrt{N}$ the eigenvectors of $H$ are localized, with localization length $\xi \propto W^2 \ll N$ and the local eigenvalue statistics are asymptotically a Poisson point process (as they are in the localized phase of Anderson's model \cite{minamiLocalFluctuationSpectrum1996}). The Green's function should obey the fractional moment condition \cref{eq:fmc}, and the diffusion constant should be zero.
		
		\item \label{it:2} For $W\gg \sqrt{N}$ the eigenvectors of $H$ are delocalized and the local eigenvalue statistics are asymptotically the same as those of GUE matrices, namely a $\mathrm{Sine}_\beta$ process with $\beta=2$. The system should exhibit ``quantum unique ergodicity'' (QUE) \cite{Bourgade2017} (see also \cite{cipolloni2021eigenstate, cipolloni2022normal, cipolloni2022rank, benigni2022fluctuations, bourgade2017eigenvector, https://doi.org/10.1002/cpa.21895} for recent QUE results for``mean-field'' models, i.e. when $W\sim N$).
		
	\end{enumerate}
	
	In the present note, we study the localization side of this problem and establish the fractional moment condition up to $W\ll N^{1/4}$ with localization length smaller than $N^{1/4}$, see \cref{thm:main localization theorem} below. Explicitly, we consider complex $nW\times nW$ random band matrices (RBM) $H$ with a band width $W$, thinking of the matrix $H$ as an $n \times n$ block matrix, with each block of size $W$, of the form \cref{eq:Hamiltonian}.
	This is the so-called \emph{Wegner W-orbital model}, introduced in \cite{wegnerOrbital1979}, which is most natural for our methods. Our main result is that the localization length of this model is bounded above by $C W^3$ with distance measured between blocks of size $W$.  Thus, the matrix exhibits localization if $n\gg W^3$, which is to say $N=nW\gg W^4$, i.e., $W\ll N^{1/4}$.   
	%In \cref{sec:Generalizations} we explain how the proof may be modified to deal with the proper RBM case (in which case the $T_j$'s are lower triangular), among other possible generalizations.

\subsection{Our approach: fluctuations, localization, and the Mermin--Wagner theorem}

%If a family of random variables is both uniformly bounded and furthermore exhibits large logarithmic fluctuations, then actually it exhibits exponential decay. This is the main mechanism of Anderson localization we shall use, applying it on the family of random variables $\Set{\norm{G(1,n;z)}^s}_{n\in\NN}$. This insight was pioneered in \cite{Schenker2009} and is explained in detail in \cref{sec:logfluca} below. 

The study of disordered systems, including RBM, shares many ideas and methods with the field of statistical mechanics.  In the present paper we follow the basic argument already presented in \cite{Schenker2009}, which combines the \emph{a-priori} bound on fractional moments of the Green's function with a lower bound on its logarithmic fluctuations. These two together yield exponential decay of the Green's function, as we will explain below.  

The fact that fluctuations lead to a decay of correlations is the key idea behind the  Mermin--Wagner theorem on the absence of continuous symmetry breaking in two-dimensional statistical mechanical models \cite{PhysRevLett.17.1133}.  In the present paper, and in \cite{Schenker2009}, a quantitative lower bound on logarithmic fluctuations, and hence on the localization length, is obtained via a collective microscopic deformation on all the random variables, so as to exhibit a macroscopic lower bound.  This argument is inspired by Dobrushin's and Shloshman's proof of the Mermin--Wagner theorem \cite{dobrushinAbsenceBreakdownContinuous1975}.  The idea to use such arguments in the context of random operators was proposed by Aizenman, as noted in \cite{Schenker2009}.  

The main idea of the proof is most conveniently explained in the context of the classical $XY$ model of statistical mechanics, which describes a circle-valued field $\theta:\Lambda=[-L,L]^2\cap\ZZ^2\to [0,2\pi)$ with Hamiltonian 
$$ H_L(\theta;\theta^{\text{bc}}) \ = -\ \sum_{x\sim y} \cos(\theta_{x}-\theta_{y}) - \sum_{(x,y)\in \partial \Lambda} \cos(\theta_{x} -\theta^{\text{bc}}_{y}), \ , $$
where $\sim$ denotes nearest neighbors, $\partial \Lambda$ is the boundary of $\Lambda$ consisting of nearest neighbor pairs with one element in $\Lambda$ and one element in $\Lambda^c$, and $\theta^{\text{bc}}$ is a \emph{boundary condition} at points immediately outside $\Lambda$.  Given the boundary condition $\theta^{\text{bc}}$, one takes $\theta$ distributed according to the Gibbs measure $\frac{1}{Z_\beta} \exp(-\beta H_L(\theta))$.  The Mermin--Wagner theorem asserts that the distribution of $\theta$ deep inside $\Lambda$ (say, in some fixed neighborhood of $0$), is insensitive to the boundary condition $\theta^{\text{bc}}$ in the limit $L\to \infty$.  This is the \emph{absence of symmetry breaking}. 

To exhibit a macroscopic fluctuation in $\theta$ (and hence rule out symmetry breaking), we make a microscopic shift  $\theta^{\pm}_{x} = \theta_{x} \pm u(x)$, with $u$ a yet unspecified function whose gradients we assume to be small throughout $\Lambda$ and zero at the boundary of $\Lambda$.  Our goal is to choose $u$ so that $u(0)$ is of order one, while the change of the Hamiltonian is small.  The simplest argument, comparing $H_L(\theta^{\pm})$ and $H_L(\theta)$ to leading order in $u$ leads to a bound that is too large to be useful.  A key observation, due to Pfister \cite{Pfister_1981_cmp/1103908962} in his proof of the Mermin--Wagner theorem, is that by combining forward and backwards shifts this leading order term is cancelled and one has
\begin{equation}\label{eq:Pfister}
    \left|-\frac{1}{2} H_L(\theta^{+}) - \frac{1}{2} H_L(\theta^{-}) + H_L(\theta)\right| \ \lesssim \   \sum_{x\sim y} (u(x) - u(y))^2 \ .
    \end{equation}
Choosing $u$ appropriately, we can obtain $u(0)=O(1)$ with $\sum_{x\sim y} (u(x) - u(y))^2 = O(1/\ln L)$. From here it is a short argument to conclude that the distribution of the $\theta_0$ is asymptotically equal to Lebesgue measure on the circle as $L\to \infty$.  
    
The previous work on localization RBM \cite{Schenker2009} relied on a fluctuation argument that did not make use of cancellations similar to those in \eqref{eq:Pfister}.  Because of the quasi-$1d$ nature of RBM models, there were still sufficient fluctuations at leading order to obtain an estimate.  In the present paper we make use of these cancellations, which enables us to improve the localization length estimate from $N^{1/7}$ to $N^{1/4}$. Compared to previous cases in which such Mermin--Wagner techniques have been used, e.g., \cite{Pfister_1981_cmp/1103908962,Milos2015}, the nature of our problem (which carries a double asymptotic parameter $n,W\to\infty$) makes the argument more delicate. We must balance the shifts (as a function of $W$) so that they on the one hand generate a macroscopic change, but on the other hand do not ``cost'' an amount which diverges as $W\to\infty$ too quickly. This cannot actually be ensured for every single matrix, since, as $n\to\infty$, it is certain that a portion of the blocks will behave badly. To deal with this problem we introduce various cut-off functions into the shift and use large deviation estimates to guarantee enough blocks behave well. Hence, the shifts we make depend on the random realization and for this reason it would be appropriate to call this an ``adaptive'' Mermin--Wagner argument. Such an adaptive approach first appeared in \cite{Richthammer2007} and subsequently also  in e.g. \cite{Milos2015,Kozma_Peled_10.1214/21-EJP639}.

	\subsection{Prior results on random band matrices}
	We now discuss the existing mathematical literature in both the localized and delocalized regimes of RBM. The first result in the localized regime \cite{Schenker2009} was a proof of localization of the eigenvectors for RBM with $W\ll N^{1/8}$; this result has then been improved to $W\ll N^{1/7}$ in \cite{10.1093/imrn/rnx145} for the Gaussian models considered here. The main result of this paper is the proof of the localization for $W\ll N^{1/4}$.
	
	The convergence of the local eigenvalue statistics for RBM in the localized regime to a Poisson point process is still open.   In \cite{brodie2020density}, Poisson statistics were proved to hold in the limit $N\to \infty$ with $W$ fixed. A recent result establishes that any non-trivial limit point of counting functions of local eigenvalue statistics is Poisson distributed \cite{hislop2022}, but convergence to a single limit with intensity given by the semi-circle law density of states has not been proved.
	
	On the delocalized side of the transition, after several results about smaller and smaller band width \cite{bao2017delocalization, erdHos2011quantum, erdHos2015altshuler, erdHos2015altshuler2, erdHos2013delocalization}, the most recent results are in \cite{https://doi.org/10.1002/cpa.21895,bourgade2019random, yang2021random}, where Bourgade, Yang, Yau, and Yin proved both delocalization of eigenvectors and convergence of local eigenvalue statistics to the corresponding GUE/GOE limiting correlation functions for RBM with $W\gg N^{3/4}$. Using supersymmetric techniques (SUSY), M. and T. Shcherbina proved the convergence of the $2$-point correlation function to the corresponding GUE counterpart for a specific model of RBM with complex entries down to the optimal $W\gg N^{1/2}$, however they do not prove delocalization of eigenvectors \cite{shcherbina2018universality, shcherbina2021universality}.
	%	We also point out that in \cite{shcherbina2018universality, shcherbina2021universality} M. and T. Shcherbina, using supersymmetric techniques (SUSY), proved the convergence of the $2$-point correlation function to the corresponding GUE counterpart for a specific model of RBM with complex entries down to the optimal $W\gg N^{1/2}$, however they do not prove delocalization of eigenvectors. 
	At the edge of the spectrum S. Sodin proved that a phase transition occurs when $W\sim N^{5/6}$ \cite{sodin2010spectral}.  In a series of works M. and T. Shcherbina \cite{Shcherbina2017} and T. Shcherbina \cite{ shcherbina2020characteristic, shcherbina2014second, shcherbina2022susy} computed, using SUSY techniques, the expectation of products of characteristics polynomials in the whole regime $1\lesssim W\lesssim N$ showing that a crossover appears at $W\sim N^{1/2}$; in \cite{shcherbina2020characteristic} even the threshold regime $W\sim N^{1/2}$ is analyzed.

	RBM have also been studied in higher dimensions. Since the focus of this paper is on $1d$ RBM, we mention only a few significant results. The limiting density of states down to arbitrary short scales has been derived in \cite{disertori2017density} and \cite{disertori2002density} for $d=2$ and $d=3$, respectively. More recently, Yang, Yau and Yin proved delocalization of eigenvectors of RBM with a band width $W\gg 1$ for $d\ge 8$ \cite{yang2021delocalization, yang2021delocalization2}.

	Before posting this manuscript but after its completion we learnt that independently Nixia Chen and Charles Smart have also obtained localization for $W\ll N^{1/4}$ for the Gaussian random band matrix model \cite{Chen_Smart2022}; the two preprints have appeared simultaneously. The two papers share the general philosophy of exhibiting exponential decay through logarithmic fluctuations. However, while we use the so-called adaptive Mermin--Wagner shift to generate logarithmic fluctuations, they, following Schenker's original paper \cite{Schenker2009}, work by analyzing the marginal distribution of a scalar degree of freedom and showing it is log-concave.

	\subsection{Organization of the paper}
	
	The rest of this paper is organized as follows. In \cref{sec:mainres}, we present our main new result, \cref{thm:main localization theorem}, concerning localization at real energies $z$ with $|z|\lesssim \sqrt{W}$, and explain how this result implies \cref{thm:first loc theorem}.
%	In \cref{sec:mainres} we separate \cref{thm:first loc theorem} into different spectral regions, and state the main new result (see \cref{thm:main localization theorem} below). 
In \cref{sec:sp} we reduce the proof of \cref{thm:main localization theorem} to a lower bound on logarithmic fluctuations of the Green's function and present our main tool to derive lower bounds on fluctuations: the Mermin--Wagner estimate. This estimate is derived then for our particular model in \cref{sec:MW}. The remaining technical estimates are delayed to \cref{sec:Calculation of the eta Jacobian,sec:the beta term bound,sec:the epsilon term bound,sec:technical}. In \cref{sec:Generalizations} we discuss generalizations to other models which may also be handled by our method.
	%\normalcolor Finally, in \cref{sec:tb} we present an alternative approach to estimate the fluctuations of $G(1,n;z)$. \corr [Do we still want to write this as well?] \nc

	%\begin{itemize}
	%	\item The absolute value of a matrix $|A|^2\equiv A^\ast A$.
	%	\item The Hilbert-Schmidt norm $\HSn{A}^2\equiv\tr\left(|A|^2\right)=\sum_{i,j=1}^W|A_{ij}|^2$.
	%\end{itemize}
	\subsection{Notations and conventions}

	For vectors $v,u\in \mathbb{C}^W$ we take the usual scalar product:
	\[
	\langle v,u\rangle\equiv \sum_{i=1}^W \overline{v_i}u_i,
	\]
	and let $\norm{v}=\sqrt{\langle v,v\rangle}$ denote the Euclidean norm.  The corresponding matrix operator norm for $A\in\mathbb{C}^{W\times W}$ is
	\[ \norm{A} = \sup_{\norm{v}=1} \norm{Av} \ . \]
	Additionally, by $|A|^2\equiv A^\ast A$ we denote the absolute value of a matrix, and by  $\HSn{A}$ we denote its Hilbert-Schmidt norm:
	\[
	\HSn{A}^2\equiv\tr\left(A^\ast A\right)=\sum_{i,j=1}^W|A_{ij}|^2\,.
	\]
	
	Finally, for positive quantities $f,g$ we write $f\lesssim g$ if there exists $C>0$, independent of any asymptotic parameter (in this paper, $n$ and $W$, with $n,W$ as in \cref{eq:Hamiltonian}) such that $f\le Cg$, and we will write $f\approx g$ if $f\lesssim g$ and $g\lesssim f$.  Furthermore, we write $f\ll g$ if there exists a small $c>0$, independent of $n$ and $W$, such that $f\le N^{-c}g$, with $N=nW$.

	\bigskip
	
	\noindent\textbf{Acknowledgements:} 
	As was already pointed out in \cite{Schenker2009}, the origin of the idea that logarithmic fluctuations combined with the a-priori bound implies localization was first proposed by Michael Aizenman, whom we thank for many useful discussions.
	
	The research of R.P. was supported by the Israel Science Foundation grants 861/15 and 1971/19 and by the European Research Council starting grant 678520 (LocalOrder) and Consolidator grant 101002733 (Transitions).
	
	The research of J.Sc. was supported by the US National Science Foundation Grants No. 1900015 and 2153946. 
	\bigskip
	
	\section{Main results}
	\label{sec:mainres}
    In this section we collect together all statements which together constitute \cref{thm:first loc theorem}. Strictly speaking the contribution of the present paper is  \cref{thm:main localization theorem}, below.
	
	Let $H$ be a random band matrix (RBM) as in \cref{eq:Hamiltonian}. Explicitly, its distribution is given by 
	\eql{
		\label{eq:density}
		\frac{1}{Z_{n,W}} \exp\left(-W \tr\left(\sum_{j=1}^{n}|V_j|^2+\sum_{j=1}^{n-1}|T_j|^2\right)\right)\dif{V_1}\dots\dif{V_n}\dif{T_1}\dots\dif{T_{n-1}}\,,
	}
	where $\dif{V_j},\dif{T_j}$ are the Lebesgue measures on Hermitian $W\times W$ matrices, and on complex $W\times W$ matrices, and $Z_{n,W}$ is a normalization constant making \cref{eq:density} into a probability distribution. Define the block Green's function $G(x,y;z)\in\Mat_{W}(\CC)$ via its matrix elements 
	\eq{ 
		G(x,y;z)_{ij} := \ip{\delta_x\otimes e_i}{(H-z\Id_N)^{-1}\,\delta_y\otimes e_j}_{\CC^N} \qquad(x,y\in\ZZ;\,i,j\in\Set{1,\dots,W})
	} with $\Set{e_j}_j$ the standard basis of $\CC^W$ and $\Set{\delta_x}_x$ the standard basis of $\CC^n$; hence $\Set{e_j\otimes\delta_x}_{j,x}$ is the standard basis of $\CC^N\cong\CC^W\otimes\CC^n$. Here the indices $x,y$ correspond to the $(x,y)$-th $W\times W$-dimensional block of $H$ and $(i,j)$ to the matrix elements within this block.
	
	Our main result is:
	
	\begin{thm}\label{thm:main localization theorem} Assume that $H$ is of the form \cref{eq:Hamiltonian} and distributed according to \cref{eq:density}. There exists an $s_0\in(0,1)$ such that for all $s\in(0,s_0)$ and $z\in\RR$ with $|z|<M$, where $M\in(0,\infty)$ and $M\lesssim \sqrt{W}$, there exist $C<\infty,\mu>0$, independent of $n$ and $W$, such that
		\eql{
			\label{eq:main fractional moment bound}
			\ex{\norm{G(x,y;z)}^s} \leq W^C \exp\left(-\mu \frac{|x-y|}{W^\sharp}\right)\qquad(x,y\in\Set{1,\dots,n})
		}
		with $\sharp=3$.
	\end{thm}
	
	We refer the reader to \cref{sec:Generalizations} for generalizations of this model which are also covered by the same method of proof.
	
	\begin{rem}[Localization at asymptotically large energies]  For $n\lesssim \ee^W$, the spectrum of $H$ is contained in a fixed compact interval with good probability.  To see this, note that the norm of $H$ is bounded by $\norm{H} \le \max_j \norm{V_j} + 2 \max_j \norm{T_j}$.  For large $W$ one has 
	$$ \prob{ \norm{V_j} > 2+t} \lesssim \ee^{-ct} \quad \text{and} \quad \prob{ \norm{T_j} > 1+t} \lesssim \ee^{-ct}  .$$  Thus with probability at least $(1-\ee^{-ct/2})^{2n}$ we have $\norm{H} \le 4+t$.  For energies $|z|>3$, one may obtain in place of \cref{eq:main fractional moment bound} the stronger estimate 
	$$\ex{\norm{G(x,y;z)}^s} \leq W^C \exp\left (-\mu\frac{|x-y|}{\log W} \right ) $$
	by using the Combes-Thomas bound \cite[Theorem 10.5]{AizenmanWarzel2016} to control the Green's function over  long intervals for which $z$ is not in the spectrum of a local restriction of $H$. The details of such an argument are quite similar to the ``finite volume criteria'' used to prove localization for Anderson's model in the Lifschitz tails regime, e.g., in \cite{Aizenman2001}, with the minor change that the \emph{a priori} bound on the Green's function (see \cref{lem:a-priori bound}) has a factor of $W^{s}$, leading to the $\log W$ localization length.   As a result, the main estimate \cref{eq:main fractional moment bound} can, in fact, be extended to all $z\in \RR$. Since our main purpose is in estimating the much longer localization length in the bulk part of the spectrum we omit further details here.
	\end{rem}
	\begin{rem}[Localization at complex energies]  We also remark that it is possible to extend the localization estimate to complex energies.  Since $\ex{\norm{G(x,y;z)}^s}$ is a subharmonic function of $z$ in the upper and lower half planes, its value at $z\notin \RR$ may be bounded by its Poisson integral real axis, which in turn may be bounded by the extension of \cref{eq:main fractional moment bound} to all real energies as outlined above \eq{
	\ex{\norm{G(x,y;E+\ii \ve)}^s} \leq \frac{1}{\pi}\int_{\lambda\in\RR}\ex{\norm{G(x,y;\lambda)}^s}\Im{\frac{1}{\lambda-E-\ii\ve}}\dif{\lambda}\,.
	}%  This is similar in spirit to \cite[Theorem 4.2]{Aizenman2001}. 
	\end{rem}
	
	The two preceding remarks, combined with the main theorem \cref{thm:main localization theorem} readily imply \cref{thm:first loc theorem}.
	
	We will in general not keep track of the polynomial $W$ dependence in the estimates below (we are mainly interested in the localization length), and so do not report on the explicit value of $C$ in \cref{eq:main fractional moment bound} which our proof yields, though in principle one may do so.

	By \cref{thm:main localization theorem} and \cite[Theorem A.1]{Aizenman2001}, we readily conclude the following
	\begin{cor}[Eigenvector localization]
		\label{cor:loc}
		Let $H$ be defined as in \cref{eq:Hamiltonian}, and let $\psi_i$, with $i\in\Set{1,\dots,N}$, be the orthonormal eigenvectors of $H$. Then there are constants $D<\infty,\nu>0$, independent of $n$ and $W$, such that for any $i,j\in\Set{1,\dots,nW}$ it holds
		\eql{\label{eq:correlator bound}
			\ex{\sum_{k\in \Set{1,\dots,N}}|\psi_k(i)\psi_k(j)|} \le W^D \exp\left(-\nu \frac{|i-j|}{W^{\sharp+1}}\right),
		}
		with $\sharp = 3$.
	\end{cor}
	
	\begin{rem}
	    As is well known \cite{AizenmanWarzel2016}, the \emph{eigenvector correlation bound} \cref{eq:correlator bound} implies directly that
	    $$ \ex{\sup_{|f|\le 1} \left|\ip{\delta_i}{f(H)\delta_j}\right| } \ \le \ W^D \exp\left(-\nu \frac{|i-j|}{W^{\sharp+1}}\right) \ ,$$
	    where the supremum is taken over all Borel measurable $f:\RR \to \RR$ satisfying $|f(x)|\le 1$ everywhere.  In particular, this implies strong dynamical localization
	    $$ \ex{\sup_{t} \left|\ip{\delta_i}{\ee^{\ii t H} \delta_j}\right| } \ \le \ W^D \exp\left(-\nu \frac{|i-j|}{W^{\sharp+1}}\right) \ .$$
	\end{rem}
	
	Note that \cref{cor:loc} implies exponential localization of all the eigenvectors of $H$ when its band width is of size $W\ll N^{1/4}$. Improving \cref{eq:main fractional moment bound} to $\sharp=1$ would amount to proving the $\sqrt{N}$-conjecture from the localization side. While the $\sqrt{N}$-conjecture is formulated using the asymptotic parameter $N$, we find it more convenient to think of the system as having $n$ slices each of size $W$, and hence we measure distances between slices. As such, the $\sqrt{N}$-conjecture is tantamount to localization up to $W \approx n$ (this is also the reason for the fact that $\sharp$ in \cref{eq:main fractional moment bound} is ``a power off'').

	% \leqnomode
	% One may consider the following models of randomness for \cref{eq:Hamiltonian}\corr [We shouldn't talk of a model for which we then don't say anything.] \nc:    \begin{equation}\label{eq:Model_1}\tag{Model 1}\parbox{30em}{{\color{orange}If $V_j=\diag(v_j^1,\dots,v_j^W)+A$ with each $v_j^k$ being an i.i.d. R.V., $A$ being the one-dimensional adjacency matrix ($-1$'s on the super- and sub-diagonals) and $T_j=\Id$ then we get the Anderson model on the strip.}} 
		%  \end{equation}
	% \begin{equation}\label{eq:Model_2}\tag{Model 2}\parbox{30em}{If $V_j$ is distributed according to $\mathrm{GUE}(W)$ and $T_j$ is distributed according to $\mathrm{Ginibre}(W)$ we get the Wegner $W$-orbital model.}
		%\end{equation}
		%  \begin{equation}\label{eq:Model_3}\tag{Model 3}\parbox{30em}{If $V_j$ has an absolutely continuous distribution w.r.t. the Lebesgue measure on $\mathrm{Herm}_W(\CC)\cong\RR^{W^2}$ with covariance matrix $\EE[|V_{ij}|^2]\cong 1/W$ and $T_j$ are allowed to be deterministic (i.e. the randomness of these variables is not used).}
			%  \end{equation}
		% \reqnomode
		%  \corr model 3 should be said to include model 2 \nc

		\section{Proof of \cref{thm:main localization theorem}}
		\label{sec:sp}
		
		In this section we explain the main steps of the proof of \cref{thm:main localization theorem}.

		\subsection{Reduction to a lower bound on logarithmic-fluctuations}
		\label{sec:logfluca}
		Using finite rank perturbation theory and the a-priori bound \cref{lem:a-priori bound}, it suffices to prove \cref{eq:main fractional moment bound} for $x=1$ and $y=n$. The energy $z$ plays very little role in our analysis and it is convenient to keep it implicit in many formulas. We thus define,
		\eql{
			\label{eq:gn}
			\calG_n := G(1,n;z)\,,\qquad X_n := \log\left(\norm{\calG_n}\right)\,.
		}
		The reason why we have defined the logarithm of the Green's function is best explained by the following lemma \cite[Proposition 3]{Schenker2009}:
		\begin{lem}\label{lem:log variance}
			Let $0<r<s<1$ and $Y\geq0$ be a random variable. Then
			\eql{
				\label{eq:logarithmic variance}
				\ex{Y^r} = \ex{Y^s}^{r/s} \exp\left(-\int_{0}^{s}f_{r,s}(q)\,\VV_{q}\left[\log\left(Y\right)\right]\dif{q}\right)\,.
			}
			where $$f_{r,s}(q) := \frac{1}{s}\min\left(\Set{r,q}\right)\left(s-\max\left(\Set{r,q}\right)\right)\qquad(q\in(0,s))\,.$$
			Here $\VV_{q}$ denotes the variance with respect to the $q$-weighted-probability measure
			%$ \dif{\mathbb{P}_{q}}\left(X\right):=\frac{\ee^{q X}}{\ex{\ee^{qX}}}\dif{\mathbb{P}\left(X\right)} $,
			\eql{
				\VV_{q}[X]=\mathbb{E}_q[\left(X-\mathbb{E}_q[X]\right)^2]\,,\qquad \mathbb{E}_q[\cdot]:=\frac{\mathbb{E}[\cdot \ee^{qX}]}{\mathbb{E}[e^{qX}]}\,.
			}
		\end{lem}
		For completeness we give a simple proof of \cref{lem:log variance} in \cref{sec:technical}. 
		
		The next crucial ingredient is the \emph{a-priori} bound on fractional moments of $\calG_n$: 
		\begin{lem}[a-priori bound]\label{lem:a-priori bound}
			For all $s\in(0,1)$ there exists $C_s<\infty$ such that
			\eql{\label{eq:a-priori bound}
				\sup_n \ex{\norm{\calG_n}^s} \leq C_s W^s\,.
			}	
		\end{lem}
		Such bounds for the resolvents of random operators are by now ``classical'' in the literature and have appeared many times elsewhere, starting from \cite{Aizenman_Molchanov_1993_cmp/1104253939}, without the supremum and for the case $n=1$ using (rank-1) finite rank perturbation theory, and then in \cite[Lemma 5]{Graf1994} for $n\geq 1$ using (rank-2) finite rank perturbation theory. For our purposes this is essentially \cite[Eq. (1.7)]{Schenker2009}, but as stated here with the optimal $W^s$ factor, can be derived from \cite[Eq. (1.7)]{doi:10.1142/S0219199717500286} (the optimal factor is not important for our proof since any polynomial factor in $W$ is negligible for \cref{thm:main localization theorem}).

		Combining \cref{lem:a-priori bound} and \cref{lem:log variance}, we see that proving \cref{thm:main localization theorem} reduces to establishing the bound 
		\eql{
			\label{eq:necessary logarithmic fluctuations bound}
			\VV_{q}\left[X_n\right] \gtrsim D_q\frac{n}{W^\sharp}
		}
		for some $q$-dependent constant $D_q\in(0,\infty)$. In fact, since the integrand in the exponential in \cref{eq:logarithmic variance} is positive, to get a lower bound we may restrict the integration to \eql{\label{eq:restriction of q to half the interval}q\in\left[\frac{s}{2},s\right]\,.}
		\subsection{The Mermin--Wagner route to a lower bound on fluctuations}
		Our main tool to establish a lower bound on fluctuations is the following
		\begin{lem}
			\label{lem:MW lower bound on fluctuations}Let $X$ be a real-valued random variable
			distributed according to the measure $\mathbb{P}$ and such that there
			are some $0<\alpha<a$ and $\ve\in\left(0,1\right)$, $\beta\in(0,\infty)$ with which
			\eql{\label{eq:MW hypothesis}
				\PP\left[\{|X|\leq\alpha\}\right] \leq \beta\sqrt{\PP\left[\{X\geq  a\}\right]\PP\left[\{X\leq -a\}\right]}+\ve\,.
			}
			Then the following lower bound holds:
			\eql{
				\ex{X^2} \geq \frac{1-\varepsilon}{1+\frac{1}{2}\beta}\alpha^{2}\,.
			}
		\end{lem}
		As phrased this lemma is inspired by \cite{Milos2015}, though it goes all the way back to the proof of the Mermin--Wagner theorem in the context of classical statistical mechanics by Pfister \cite{Pfister_1981_cmp/1103908962}. Its simple proof is postponed to \cref{sec:technical}.

		We define the centered observable $ \overline{X_n} := X_n - \EE_q\left[X_n\right]$ with which our main goal for the rest of the paper is thus to prove
		\begin{prop}\label{prop:main MW estimate} Let $\xi >0$ and $s\in(0,1)$ be fixed parameters independent of $n$ and $W$.  For the variables defined in \cref{eq:gn} above, if $n$ and $W$ are sufficinetly large and 
		\eql{\label{eq:assumption which if false proves the main theorem} \ex{\ee^{\frac{s}{2}X_n}} \ge \ee^{-\xi n} \ ,
		}
		then for all $q\in(s/2,s)$ we have
			\eql{\label{eq:The MW bound} \PP_q\left[\Set{|\overline{X_n}| \leq \alpha }\right] \leq \beta \sqrt{\PP_q\left[\Set{\overline{X_n} \geq 2\alpha}\right] \PP_q\left[\Set{\overline{X_n} \leq -2\alpha}\right]} + \ve} where $\beta \approx 1$, $\ve < 1/2$, and \eq{\alpha \approx \sqrt{\frac{n}{W^\sharp}} \ }
			with $\sharp\ge 3$.
		\end{prop}
		
		We now discuss in more detail how the main result \cref{thm:main localization theorem} follows from \cref{prop:main MW estimate}.  First note that if \cref{eq:assumption which if false proves the main theorem} fails, then \cref{eq:main fractional moment bound} holds with $\sharp=0$ and $C=1$, which is stronger than the bound we seek to prove.  So, we may assume that \cref{eq:assumption which if false proves the main theorem} holds without loss of generality. Then \cref{prop:main MW estimate} and \cref{lem:MW lower bound on fluctuations} together show that \cref{eq:necessary logarithmic fluctuations bound} holds.  Using this estimate on the right hand side of \cref{eq:logarithmic variance} we find that 
		$$ \EE[ \ee^{r X_n} ] \le \EE[\ee^{sX_n}]^{r/s} \exp\left(- c \frac{n}{W^\sharp} \right ) \ , $$
		which implies \cref{eq:main fractional moment bound} after using the \emph{a priori} bound \cref{lem:a-priori bound} to estimate $\EE[\ee^{sX_n}]$.   Thus, \cref{thm:main localization theorem} will be proved once we obtain \cref{prop:main MW estimate}.
		
		\section{The proof of the main estimate, \cref{prop:main MW estimate}}
		\label{sec:MW}
		
		In this section we present the proof of our main technical result \cref{prop:main MW estimate}. Its proof is divided into two step: first in \cref{sec:fac} we perform a change of variables for the $V$'s in the density \cref{eq:density}, which yields a convenient factorization of the resolvent $\mathcal{G}_n$; then in \cref{sec:MWshift} we describe the Mermin--Wagner shift argument (inspired by \cite{Milos2015}) which yields \cref{eq:The MW bound}.

		\subsection{Factorization of $\calG_n$ and replacing $V$ with $\Gamma$}
		\label{sec:fac}
		
		Before starting with the actual proof of \cref{eq:The MW bound}, it is convenient to first factorize $\calG_n$ (this is performed closely along the lines of \cite[Section 3]{Schenker2009}):
		\begin{lem}
			Let $\mathcal{G}_n$ be defined in \cref{eq:gn}, then it holds that
			\begin{align}\label{eq:Green's function factorization} \calG_n = \Gamma_{1}^{-1}T_{1}^{\ast}\Gamma_{2}^{-1}T_{2}^{\ast}\dots\Gamma_{n-1}^{-1}T_{n-1}^{\ast}\Gamma_{n}^{-1}, \end{align} where \begin{align}\label{eq:what Gamma equals to which we actually use}\Gamma_1 := V_1-z\Id\,;\qquad\Gamma_j := V_{j}-z\Id-T_{j-1}\Gamma_{j-1}^{-1}T_{j-1}^{\ast} \qquad (j=2,\dots,n)\,.\end{align}
		\end{lem}
		\begin{proof}
			By the resolvent formula, we have $$ G_{\left[1,n\right]}\left(1,n;z\right)=G_{\left[1,n-1\right]}\left(1,n-1;z\right)T_{n-1}^{\ast}G_{[1,n]}(n,n;z) $$ where $H_{\left[x,y\right]}$ denotes the matrix $H$ restricted to be non-zero only between slices $x$ and $y$, and $G_{[x,y]}$ denotes its resolvent. Iterating this identity $n$ times, and defining \begin{align}\label{eq:def of Gamma}\Gamma_j := G_{\left[1,j\right]}\left(j,j;z\right)^{-1}\qquad(j=1,\dots,n),\end{align} we conclude \cref{eq:Green's function factorization}. Finally, the fact that $\Gamma_j$ can be written in the form \cref{eq:what Gamma equals to which we actually use} readily follows by the Schur complement formula.
		\end{proof}
		
		%It is equivalent to bound either $\calG_n$ or $\calG_n^\ast$, and also, we get rid of the odd piece at the end, $\Gamma_n^{-1}$, at the cost of deteriorating $s$ and getting further polynomial factors in $W$, so that with mild abuse of notation we redefine \eq{  \calG_n\mapsto T_{n}\Gamma_{n}^{-1}\dots T_1\Gamma_1^{-1}\,. } 
		
		Since $\Gamma_j$ does not depend on $V_k$ for $k>j$, a change of variables $$ V_j \mapsto \Gamma_j \qquad(j=1,\dots,n)$$ is triangular and its Jacobian has determinant equal to $1$. After the change of variables, we thus obtain the density 
		\eq{ \frac{1}{Z_{n,W}} \exp\left(-W E(\Gamma,T) \right)\dif{\Gamma_1}\dots\dif{\Gamma_n}\dif{T_1}\dots\dif{T_{n-1}} 
		} where $\dif{\Gamma_j}$ is the Lebesgue measure on $\mathrm{Herm}_W(\CC)$ and $\dif{T_j}$ is the Lebesgue measure on $\Mat_W(\CC)$, and where we define the ``energy'' functional \eq{ E(\Gamma,T) \equiv \tr\left(|\Gamma_1+z\Id|^2+\sum_{j=2}^{n}|\Gamma_j+z\Id+T_{j-1} \Gamma_{j-1}^{-1}T_{j-1}^\ast|^2+\sum_{j=1}^{n-1}|T_j|^2\right)\,.}
		
		\subsection{Plus-minus collective deformations}
		\label{sec:MWshift}
		
		The main idea behind establishing a bound such as \cref{eq:The MW bound} is to perform a collective change of variables that has a minimal effect (quantified in our case by $\beta$) on the measure whereas the cumulative effect on the observable (i.e., $\chi_{[-\alpha,\alpha]} (X_n) $) is significant. This is done little by little, spread ``across the volume.'' Furthermore, the deformation is done simultaneously in two directions, designed precisely so that the linear term between the two cancels.
		
		There are many possible choices for the deformation; below is a relatively simple one which yields $\sharp = 3$. Part of the simplicity comes from the fact we are using the randomness of both the hopping terms $T$ and the onsite potential $V$. It is reasonable to guess that one could obtain the same result (with considerable complications to the proof) by using the randomness of only one of $T$ or $V$.
		
		We define the following transformation on the set of hopping matrices. Let \eq{\vf:\left[0,\infty\right)\to[0,1]} be a smooth function satisfying
		\eq{\chi_{[0,K]} \le \vf \le \chi_{[0,2K]} \quad \text{and} \quad |\vf'| \leq \chi_{[0,2K]}  \ ,}
		with $K>1$ a constant independent of $n$ and $W$ to be fixed below.
		Next,we define the shifts
		\eql{
			\label{eq:pm transformation}
			T_j^\pm := \exp\left(\pm\delta F_j\right) T_j \qquad(j=1,\dots,n-1)\,,
		}
		where $\delta>0$ is a parameter, which we choose depending on $n$ and $W$ below, and $F_j$ is a number given by
		\eql{\label{eq:F_j}
			F_j := \vf\left(\frac{\HSn{T_j}^2}{W}\right)\vf\left(\frac{\HSn{V_{j+1}}^2}{W^2}\right)\vf\left(\frac{\HSn{\Gamma_{j+1}}^2}{W^2}\right)\vf\left(\frac{\HSn{\Gamma_{j}^{-1}}^2}{W^2}\right),
		} 
		for any $j\in\Set{1,\dots,n}$. Since we consider $T$ and $\Gamma$ as integration variables rather than $T$ and $V$, in \cref{eq:F_j} we use $V_{j+1}$ only for convenience of notation; it should be understood as a function of both $T$ and $\Gamma$, i.e., \eq{ V_{j+1}=V_{j+1}(\Gamma_{j+1},\Gamma_j^{-1},T_j) = \Gamma_{j+1}+z\Id+T_j \Gamma_j^{-1}T_j^\ast\,. } Clearly, conditioned on the $\Gamma$ variables, $F_j$ is a function of $T_j$ alone (and no other $T_k$ for $k\neq j$), so that this transformation is diagonal in the variable $j$. We point out that this choice for \cref{eq:F_j} is what dictates $\sharp\geq3$, as will become clear below.

		It is useful to consider this change of variables abstractly using the maps $\eta^\pm:\Mat_W(\CC)\to\Mat_W(\CC)$ defined by
		%		Later on we will also consider this change of variables in more abstract notation as  
        \eql{ \notag \eta^\pm (A)  &\equiv \eta^\pm(A;G,\tilde{G}) \\  &= \exp\left(\pm \delta \gamma(G,\tilde{G})  \vf\left(\frac{\HSn{A}^2}{W}\right)\vf\left(\frac{\HSn{G+z\Id+A\tilde{G}A^\ast}^2}{W^2}\right)\right)A \label{eq:abstract eta}} where $G,\tilde{G}\in \Herm_W(\CC)$ and $\gamma(G,\tilde{G}):=\vf\left(\frac{\HSn{G}^2}{W^2}\right)\vf\left(\frac{\HSn{\tilde{G}}^2}{W^2}\right)$; $z\in\RR$ with $|z|\lesssim\sqrt{W}$. We will need the following lemma (whose proof is postponed to \cref{sec:Calculation of the eta Jacobian}).
        \begin{lem}\label{lem:eta is injective}
            For $\delta W\ll 1 $ the maps $\eta^\sigma:\Mat_W(\CC)\to\Mat_W(\CC)$ are injective for any choice of $G,\tilde{G}\in \Herm_W(\CC)$.
		\end{lem} \noindent In \cref{eq:pm transformation} $\eta^\pm$ are used with
        $G:=\Gamma_{j+1},\tilde{G}:=\Gamma_j^{-1}\in\mathrm{Herm}_W(\CC)$.  By a slight abuse of notation, we will use $\eta^\sigma$, $\sigma=\pm$, to  also denote the corresponding product maps 
        \eq{\{ T_j\}_{j=1}^{n-1} \ \mapsto \ \{ \eta^\pm(T_j;\Gamma_{j+1},\Gamma_j^{-1})\}_{j=1}^{n-1}}
        on $\Mat_W(\CC)^{n-1}$.

		With this transformation, $X_n$ transforms as \eql{\label{eq:the pm transformation on X_n} X_n^\pm = X_n \pm \delta F } with $F := \sum_{j=1}^n F_j$. Another trivial but useful fact is that \eql{\label{eq:Breaking X_n into pm} X_n = \frac12 X_n^+ + \frac12 X_n^-\,. }
		By \cref{eq:the pm transformation on X_n} we get \eql{\label{eq:deformation of the set} S_\alpha := \Set{\overline{X_n}\in B_\alpha(0)} = \Set{X_n^\pm \in B_\alpha(\EE_q\left[X_n\right]\pm\delta F) } = \Set{\overline{X_n^\pm} \in B_\alpha(\pm\delta F) }} where $B_\alpha(t):= [t-\alpha,t+\alpha]$ and $\overline{X_n^\pm} := X_n^\pm - \EE_q\left[X_n\right]$.
		
		For any event $M$, we have 
		\eq{
			Q(M) \ &:= \ \PP_q\left[S_\alpha\cap M\right]\ex{\ee^{qX_n}} Z_{n,W} \\ 
			&= \int_{(\Gamma,T)\in M}\ee^{q X_n-WE(\Gamma,T)}\chi_{\Set{\overline{X_n}\in B_\alpha(0)}}\dif{\Gamma}\dif{T} \\
			&= \int_{(\Gamma,T)\in M}\ee^{+W\calR(\Gamma,T)}\prod_{\sigma\in\Set{\pm}}\sqrt{\ee^{q X_n^\sigma-WE(\Gamma,T^\sigma)}\chi_{\Set{\overline{X_n^\sigma} \in B_\alpha(\sigma\delta F) }}}\dif{\Gamma}\dif{T}  \ ,}
	    where we have defined the ``remainder term'' for the energy functional to be \eql{\label{eq:The energy remainder}
			\mathcal{R}(\Gamma,T) := \frac12 E(\Gamma,T^+) + \frac12 E(\Gamma,T^-) - E(\Gamma,T)\,.
		}
		Next we define the Jacobian determinant $$J^\pm := \left|\det\left(\calD \eta^\pm\right)\right| = \frac{1}{\left|\det\left(\calD \left(\eta^\pm\right)^{-1}\right)\right|\circ\eta^\pm}\,,$$ with $\calD$ being the total differential of the map (the Fréchet derivative), and the second equality follows from the chain rule for $\calD$ and $f^{-1}\circ f = \Id$. We thus have 
		\eq{
			Q(M)
		&= \int_{(\Gamma,T)\in M}\ee^{+W\calR(\Gamma,T)}\prod_{\sigma\in\Set{\pm}}\sqrt{\ee^{q X_n^\sigma-WE(\Gamma,T^\sigma)}\chi_{\Set{\overline{X_n^\sigma} \in B_\alpha(\sigma\delta F) }}J^\sigma(T) \frac{1}{J^\sigma(T)}}\dif{\Gamma}\dif{T} \\
			&\leq 
			\norm{\ee^{W\calR}\prod_{\sigma\in\Set{\pm}}\frac{1}{\sqrt{J^\sigma}}}_{L^\infty(M)}
			\prod_{\sigma\in\Set{\pm}}\sqrt{\int_M \ee^{q X_n^\sigma -WE (\Gamma,T^\sigma)}\chi_{\Set{\overline{X_n^\sigma} \in  B_{\alpha}(\sigma \delta F)}} J^\sigma \dif{\Gamma}\dif{T}} \ . } 
		where we have used the Cauchy-Schwarz inequality after an $L^\infty$ bound on the pre-factor.  
		
		At this point we would like to apply the change of variables formula to the two integrals on the right hand side, however the characteristic function depends in a complicated way on the variables $\{T_j\}_{j=1}^{n-1}$ through the function $F$, making it difficult to determine the domain of integration.  To circumvent this difficulty we note that on the event $M$, 
		$$ B_\alpha(\pm\delta F)\subseteq\pm(\delta\inf_M F-\alpha, \infty)  $$ where $-(a,\infty)=(-\infty,-a)$.  Thus
		\eq{
			Q(M) & \le \norm{\ee^{W\calR}\prod_{\sigma\in\Set{\pm}}\frac{1}{\sqrt{J^\sigma}}}_{L^\infty(M)}
			\prod_{\sigma\in\Set{\pm}}\sqrt{\int_M \ee^{q X_n^\sigma -WE (\Gamma,T^\sigma)}\chi_{\Set{\overline{X_n^\sigma} \in  \pm(\delta\inf_M F-\alpha, \infty)}} J^\sigma \dif{\Gamma}\dif{T}} \\
			&= 
			\norm{\ee^{W\calR}\prod_{\sigma\in\Set{\pm}}\frac{1}{\sqrt{J^\sigma}}}_{L^\infty(M)}
			\prod_{\sigma\in\Set{\pm}}\sqrt{\int_M\ee^{q X_n-WE(\Gamma,T)}\chi_{\Set{\overline{X_n} \in \pm(\delta\inf_M F-\alpha, \infty) }} \dif{\Gamma} \dif{T}}\,.}
		where we have applied the change of variables formula $\int f\circ \eta^\sigma J^\sigma d\Gamma dT = \int f \dif{\Gamma} \dif{T}$, which is valid by \cref{lem:eta is injective} provided we choose $\delta$ so that $\delta W \to 0$ (which we will).  
		Dividing by $\EE\left[\ee^{qX_n}\right]  Z_{n,W}$ we obtain the estimate
		\begin{equation}\label{eq:P(SM)} 
		\PP_q\left[S_\alpha\cap M\right] \le  \norm{\ee^{W\calR}\prod_{\sigma\in\Set{\pm}}\frac{1}{\sqrt{J^\sigma}}}_{L^\infty(M)}  \prod_{\sigma\in\Set{\pm}}\sqrt{\PP_q\left[\Set{\overline{X_n}\in \sigma(\delta\inf_M F-\alpha, \infty)}\right]}\end{equation}
		This concludes the bound on the $q$-probability of the set $S_\alpha\cap M$.
		
		We now estimate the $q$-probability of the complementary set, $S_\alpha\cap M^c$:  
		\eq{\PP_q\left[S_\alpha\cap M^c\right] =\frac{\ex{\ee^{q X_n}\chi_{S_\alpha \cap M^c}}}{\ex{\ee^{q X_n}}} \ .}
		To bound the denominator, we use the assumption \cref{eq:assumption which if false proves the main theorem} to obtain
		\eq{ \ex{\ee^{q X_n}} \ge \ex{\ee^{\frac{s}{2}X_n}}^{\frac{2q}{s}} \ge \ee^{2\xi n} ,}
		where in the middle step we have used Jensen's inequality, which is valid because $\frac{s}{2}<q<s$. For the numerator, we use Cauchy-Schwarz and Jensen's inequality again to obtain
		\eq{\ex{\ee^{q X_n}\chi_{S_\alpha \cap M^c}} \le \sqrt{\ex{\ee^{2qX_n}}}\sqrt{\PP\left[M^c\right]} \le W^{\frac12\frac{q}{s}C_{2s}}\sqrt{\PP\left[M^c\right]} \ , }
		where we have applied the \emph{a priori} bound \cref{lem:a-priori bound} in the last step. 
		Putting these estimates together we find that
		\eql{\label{eq:P(SMc)}
			\PP_q\left[S_\alpha\cap M^c\right] = \leq W^{\frac12\frac{q}{s}C_{2s}}\ee^{+2 \xi n}\sqrt{\PP\left[M^c\right]}
		}

		Combining \cref{eq:P(SM),eq:P(SMc)} we find \eq{
			\PP_q\left[\Set{\overline{X_n}\in B_\alpha(0)}\right] \leq & \norm{\ee^{W\calR}\prod_{\sigma\in\Set{\pm}}\frac{1}{\sqrt{J^\sigma}}}_{L^\infty(M)} \prod_{\sigma\in\Set{\pm}}\sqrt{\PP_q\left[\Set{\overline{X_n}\in \sigma(\delta\inf_M F-\alpha, \infty)
		    }\right]}\\&+W^{\frac12\frac{q}{s}C_{2s}}\ee^{+2\xi n}\sqrt{\PP\left[M^c\right]}\,.
		}
		Note that this last estimate is of the form \cref{eq:The MW bound}, yielding $\VV_q\left[S_\alpha\right]\gtrsim \frac{n}{W^\sharp}$, \emph{if} we can find an event $M$ for which the following three conditions are simultaneously satisfied:
		\begin{enumerate}
			\item $\inf_M F \geq \phi n $ for some $\phi \in (0,1]$ independent of $n,W$. 
			\item The $\beta$ term is of order $1$, i.e., \eql{\label{eq:beta term condition}\norm{\ee^{W\calR}\prod_{\sigma\in\Set{\pm}}\frac{1}{\sqrt{J^\sigma}}}_{L^\infty(M)} \lesssim 1\,.}
			\item The $\ve$ term is smaller than $1/2$, i.e.,
			\eql{\label{eq:epsilon term condition} W^{\frac12\frac{q}{s}C_{2s}}\ee^{+2\xi n}\sqrt{\PP\left[M^c\right]} < 1/2 \ . } 
		\end{enumerate}
		
		Such an event \emph{does} indeed exist. Fix some \eql{\label{eq:choice of phi}\phi\in\left(0,\frac{1}{6}\right)} and define 
		\eql{\label{eq:the event M}
			M_\phi := \Set{(\Gamma,T)| F \geq  \phi n } \ ,
		}
		which clearly fulfills the first condition. The requirement that $\delta \phi n -\alpha =2\alpha$ fixes $\delta$: \eql{\label{eq:delta}\delta := \frac{3}{\phi n} \alpha = \frac{3}{\phi}\frac{1}{\sqrt{n W^\sharp}}\,.}
		The proof of \cref{thm:main localization theorem} will hence be completed with the demonstration of \cref{eq:beta term condition,eq:epsilon term condition} for the specific choice $M=M_\phi$, which is the contents of \cref{sec:the beta term bound,sec:the epsilon term bound} respectively. In the proof of \cref{eq:beta term condition} we will see that it is necessary to take $\sharp\geq3$. We start, however, with the calculation of the Jacobian associated to the change of variables used above.

		\section{The derivative of the change of variables, $\calD\eta$}
		\label{sec:Calculation of the eta Jacobian}
		
		In this section we explicitly calculate the Jacobian of the map $\eta$ defined in \cref{eq:abstract eta} and prove \cref{lem:eta is injective}. Note that $\eta$ is in fact \emph{not} $\CC$-differentiable (since, e.g., $\Mat_W(\CC)\ni A\mapsto\HSn{A}^2$ is not). This does not matter, since for the change of variable $T\mapsto T^\pm$ we are concerned with above $\RR$-differentiability suffices. So we shall use the (obvious) isomorphism $\Mat_W(\CC)\cong\RR^{2W^2}$ when calculating the Jacobian. %Hence, choosing any bijection $b:\Set{1,\dots,W^2}\to \Set{1,\dots,W}^2$, we consider $\eta$ as a map $\RR^{2W^2}\to\RR^{2W^2}$  using the isomorphism \eq{\Mat_W(\CC)\ni A \mapsto \left(A_{b(l)}\right)_{l=1}^{W^2}\mapsto \left(\left(\Re{A_{b(l)}}\right)_{l=1}^{W^2},\left(\Im{A_{b(l)}}\right)_{l=1}^{W^2}\right) \in \RR^{2W^2}\,.}
		
				We use the notation for the (Fréchet, or total) derivative which may be characterized as the linear approximation, i.e., \eql{\eta(A+\ve B) = \eta(A)+\ve \, \left(\calD\eta\right)_A \, B + \mathrm{o}\left(\ve\norm{B}\right)\qquad(A,B\in\Mat_W(\CC);\ve\to0^+)\,.
		}
		Thus schematically, from the definition \cref{eq:abstract eta} and the notation $$ \Phi(A) := \gamma  \vf\left(\frac{\HSn{A}^2}{W}\right)\vf\left(\frac{\HSn{G+z\Id+A\tilde{G}A^\ast}^2}{W^2}\right)\qquad(A\in\Mat_W(\CC)), $$ using the product rule, we have \eq{
			\left(\calD \eta\right)_A B &= \exp\left(\sigma\delta \Phi(A)\right)B + \exp\left(\sigma\delta \Phi(A)\right)\left(\sigma\delta\left(\calD \Phi\right)_A B\right)A \\
			&= \exp\left(\sigma\delta \Phi(A)\right)\left(B+ \sigma \delta \left(\left(\calD \Phi\right)_A B\right) A\right)
		}
		where it should be noted that since $\Phi$ is scalar-valued, $\left(\calD \Phi\right)_A:\RR^{2W^2}\to\RR$ is a $1\times 2W^2$ matrix, i.e., $\left(\calD \Phi\right)_A B$ is just a number. By the Riesz representation theorem there exists some vector $Q_A\in\RR^{2W^2}$ such that $ \left(\calD \Phi\right)_A B = \ip{Q_A}{B}_{\RR^{2W^2}}$. With this notation we all together have \eql{\label{eq:Jacobian of eta in terms of Jacobian of Phi} \left(\calD\eta\right)_A = \exp\left(\sigma\delta \Phi(A)\right)\left(\Id_{\RR^{2W^2}}+\sigma\delta A\otimes_{\RR^{2W^2}} Q_A^\ast\right)\,. } We now proceed to calculate $Q_A$.
		
		If $A\in\Mat_W(\CC)$, we have two matrices $A^R,A^I\in\Mat_W(\RR)$ defined via their elements $$(A^R)_{ij} := \Re{A_{ij}},\qquad (A^I)_{ij} := \Im{A_{ij}}$$ which yield then a vector $(A^R,A^I)\in\RR^{2W^2}$. In terms of this notation,
		\begin{enumerate}
			\item If $f:\Mat_W(\CC)\to\RR$ is defined by $f(A) = \HSn{A}^2$ then \eq{\frac12\left(\calD f\right)_A B = \Re{\HSip{A}{B}} = \ip{A^R}{B^R}_{\RR^{W^2}}+\ip{A^I}{B^I}_{\RR^{W^2}} =  \ip{(A^R,A^I)}{(B^R,B^I)}_{\RR^{2W^2}}\,.}
			\item If $g:\Mat_W(\CC)\to\RR$ is defined by $g(A) = \HSn{G_z+A\tilde{G}A^\ast}^2$ (with $G_z:=G+z\Id$, but we write $G$ for $G_z$ in this calculation for simplicity of notation) then \eq{\frac12\left(\calD g\right)_A B &=\Re{\HSip{G+A\tilde{G}A^\ast}{B\tilde{G}A^\ast+A\tilde{G}B^\ast}}\\
			%&=\Re{\HSip{\left(G+A\tilde{G}A^\ast\right)\left(\tilde{G}A^\ast\right)^\ast}{B}}+\Re{\HSip{\left(A\tilde{G}\right)^\ast\left(G+A\tilde{G}A^\ast\right)}{B^\ast}}\\
			%&=\Re{\HSip{\left(G+A\tilde{G}A^\ast\right)A\tilde{G}}{B}}+\Re{\HSip{\tilde{G}^\ast A^\ast\left(G+A\tilde{G}A^\ast\right)}{B^\ast}}\\
			%&=\Re{\HSip{GA\tilde{G}+A\tilde{G}|A|^2\tilde{G}}{B}}+\Re{\HSip{\tilde{G}^\ast A^\ast G+\tilde{G}^\ast |A|^2\tilde{G}A^\ast}{B^\ast}}\\
			&=\ip{((GA\tilde{G}+A\tilde{G}|A|^2\tilde{G})^R,(GA\tilde{G}+A\tilde{G}|A|^2\tilde{G})^I)}{(B^R,B^I)}_{\RR^{2W^2}}+\\
			&\quad+\ip{((\tilde{G}^\ast A^\ast G+\tilde{G}^\ast |A|^2\tilde{G}A^\ast)^{T,R},-(\tilde{G}^\ast A^\ast G+\tilde{G}^\ast |A|^2\tilde{G}A^\ast)^{T,I})}{(B^R,B^I)}_{\RR^{2W^2}}
			} where by $T$ we mean the transpose. 
		\end{enumerate}
		Thus since $\Phi(A) = \gamma \vf\left(\frac{f(A)}{W}\right)\vf\left(\frac{g(A)}{W^2}\right)$, we find \eq{
			\left(\calD \Phi\right)_A  &= \gamma \vf'\left(\frac{f(A)}{W}\right)\vf\left(\frac{g(A)}{W^2}\right)\frac{1}{W}\left(\calD f\right)_A  +\gamma \vf\left(\frac{f(A)}{W}\right)\vf'\left(\frac{g(A)}{W^2}\right)\frac{1}{W^2}\left(\calD g\right)_A \,.
		} Collecting everything together we can read off $Q_A$ (written now, for convenience, back as an element in $\Mat_W(\CC)$):
		\eql{\label{eq:Q_A} 
		Q_A &= \gamma \vf'\left(\frac{f(A)}{W}\right)\vf\left(\frac{g(A)}{W^2}\right)\frac{2}{W}A+\\
			&\qquad+\gamma \vf\left(\frac{f(A)}{W}\right)\vf'\left(\frac{g(A)}{W^2}\right)\frac{4}{W^2}\left((G+z\Id)A\tilde{G}+A\tilde{G}|A|^2\tilde{G}\right)\,.
		}
		
		We now prove the following basic estimate, which will be used within the proof of \cref{lem:eta is injective} (which is presented at the end of this section) and for other estimates later on.
		
		\begin{lem}\label{lem:technical estimate on A and Q_A}
		    For $A$ and $Q_A$ as defined above, 
		    \eql{\label{eq:technical estimate on A and Q_A}
		        \HSn{A}\HSn{Q_A} \lesssim W\,.
		    }
		\end{lem}
		\begin{proof}
		    We will use the facts that $|\vf'|\leq\chi_{[0,2K]}$ and that $\gamma$ controls the size of $G$ and $\tilde{G}$ (which follows by \cref{eq:F_j}). Thus, in all of these estimates, due to the various factors of $\vf$, we can always assume that: $\HSn{A}^2 \lesssim  W$, $\HSn{G} \lesssim W$, $\HSn{\tilde{G}} \lesssim W$, as well as $\HSn{G+z\Id+A\widetilde{G}A^*} \lesssim W$.
		    
		    Using the triangle inequality and submultiplicativity of the Hilbert-Schmidt norm, under the assumption that all these expressions are multiplied by appropriate factors of $\vf$, for $\HSn{A}\HSn{Q_A}$ we find:
		    \begin{enumerate}
		        \item $\frac{1}{W}\HSn{A}^2 \lesssim 1$.
		        \item $\frac{1}{W^2}\HSn{A} \HSn{GA\tilde{G}} \lesssim \frac{1}{W^2} W^3 = W$ and $\frac{1}{W^2}\HSn{A} \HSn{z A\tilde{G}} \lesssim |z|\frac{1}{W^2} W^2 = |z|$. This last error is negligible since we are assuming $|z|\lesssim \sqrt{W}$.
		        \item For the last term, we replace $A\tilde{G} A^\ast$ by $G+z\Id+A\widetilde{G}A^* - G -z\Id$, whose Hilbert-Schmidt norm is bounded by $W$, and so, $\frac{1}{W^2}\HSn{A}\HSn{A\tilde{G}|A|^2\tilde{G}}\lesssim \frac{1}{W^2}W^3 = W$.
		    \end{enumerate}
		\end{proof}

	We conclude this section with the proof of \cref{lem:eta is injective}.
	
		\begin{proof}[Proof of \cref{lem:eta is injective}]
			% THIS IS WRONG: Clearly it suffices to show $\eta$ is invertible and using the inverse function theorem it then suffices to show $(\calD \eta)_A$ is an invertible matrix for every choice of $A$.
			
			If $\Id_{\RR^{W^2}}-\eta$ is a contraction then $\eta$ is injective, which would be guaranteed (using the mean-value theorem e.g. on $[0,1]\ni t\mapsto\eta(t A+(1-t)B)$) if $\norm{\Id_{\RR^{W^2}}-\left(\calD \eta\right)_A}<1$ for all $A\in\RR^{2W^2}$. 
			
			Starting from \cref{eq:Jacobian of eta in terms of Jacobian of Phi}, we see that to prove this it suffices to show that \eq{ \norm{(1-\ee^{\sigma\delta\Phi(A)})\Id_{\RR^{2W^2}}-\ee^{\sigma\delta\Phi(A)}\sigma \delta A \otimes_{\RR^{2W^2}} Q_A^\ast}_{\text{operator norm of $\RR^{2W^2}$}} < 1\,. } 
	        Since $\delta \sim \frac{1}{\sqrt{n W^\sharp}}$, using the triangle inequality, it suffices to concerntrate only on the second term. We remark that for any $u,v\in\RR^{m}$, $\norm{u\otimes v^\ast}_{\text{op. norm of $\RR^{m}$}}\leq \norm{u}_{\RR^{m}}\norm{v}_{\RR^{m}}$ and that the Euclidean norm on $\RR^{2W^2}$ is precisely the Hilbert-Schmidt norm on $\Mat_W(\CC)$, i.e. $$ \norm{(A^R,A^I)}_{\RR^{2W^2}} = \HSn{A}\,. $$
			
			It is thus enough to ensure that \eql{\label{eq:final requirement to make sure eta is invertible} 2\delta \HSn{A}\HSn{Q_A} < 1/2\,. } Using \cref{lem:technical estimate on A and Q_A} to estimate $\HSn{A}\HSn{Q_A}\lesssim W$ and the fact that $\delta\sim n^{-1/2}W^{-3/2}$, we have together $$ \delta \HSn{A}\HSn{Q_A} \lesssim n^{-1/2}W^{-1/2}\,. $$ so we can certainly fulfill \cref{eq:final requirement to make sure eta is invertible}. This concludes the proof of this lemma.
		\end{proof}

		\section{The $\beta$ bound: proof of \cref{eq:beta term condition}}
		\label{sec:the beta term bound}

		The goal of this section is to establish \cref{eq:beta term condition}. We do this separately for $e^{W\calR}$ and for the product of the two Jacobians.
		\subsection{The remainder term $\calR$}
		We divide the remainder term $\cal R$ into the part that depends on $T$ directly and the part that depends on $T$ through $\Gamma$: 
		\eq{
			\calR^I_j &:= \frac12\HSn{T_j^+}^2+\frac12\HSn{T_j^-}^2-\HSn{T_j}^2,\\
			\calR^{II}_j &:= \frac12\HSn{V_j^+}^2+\frac12\HSn{V_j^-}^2-\HSn{V_j}^2,
		}
		where \eq{V_j \equiv  \Gamma_{j}+z\Id+T_{j-1} \Gamma_{j-1}^{-1}T_{j-1}^\ast\,.}
		
		For $\calR^I_j$, a Taylor expansion yields \eq{ |T_j^\pm|^2 = |T_j|^2\left(1 \pm 2 \delta F_j + 4 \delta^2 F_j^2  + \calO(\delta^3F_j)\right)\,. } 
		Applying \eq{ \frac12\HSn{A+ B}^2+\frac12\HSn{A- B}^2 - \HSn{A}^2=\HSn{B}^2 } we find \eq{ \calR^I_j \leq  \delta^2F_j\HSn{T_j}^2\left( 4 +\calO(\delta)\right)\lesssim \delta^2 K W  }
		where in the last step we have used \cref{eq:F_j}. We have $n$ terms summed in $\calR$, which is also multiplied by $W$ outside, so that inside the exponent we have all together for this term, $W^2 n \delta^2 \sim W^{2-\sharp}$ and $\sharp = 3$ even makes the $\delta^2$ term vanish in the limit $W\to\infty$, since by \cref{eq:delta}, $\delta\sim\frac{1}{\sqrt{n W^\sharp}}$.
		
		For $\calR^{II}_j$, defining $A,B,C$ through (omitting terms of order $\delta^3$ and even higher powers as above) \eq{ \left|\Gamma_{j}+z\Id+T_{j-1}^{\pm}\Gamma_{j-1}\left(T_{j-1}^{\pm}\right)^{\ast}\right|^{2} =: \left|A\pm\delta B+\delta^{2}C\right|^{2}} yields \eq{ \calR^{II}_j = \delta^{2}\norm{B}_{\text{HS}}^{2}+\delta^{2}\Re{ \left\langle C,A\right\rangle _{\text{HS}}}  } where we have used \eq{ \frac12\HSn{A+ B+C}^2 + \frac12\HSn{A- B+ C}^2 - \HSn{A}^2=\HSn{B}^2+2\Re{\HSip{C}{A}}\,. }
		
		Hence we need to estimate, first: 
		\eq{ \HSn{B}^2 &= 4 F_{j-1}^2 \HSn{T_{j-1}\Gamma_{j-1}^{-1}T_{j-1}^\ast}^2 \\
			&= 4 F_{j-1}^2 \HSn{V_j - \Gamma_j - z\Id}^2 \\
			&\lesssim F_{j-1}^2 \left(\HSn{V_j}^2+\HSn{\Gamma_j}^2+z^2W\right) \\
			&\lesssim KW^2+z^2 W
		} where in the last step we've used \cref{eq:F_j} and here it is clear that $\sharp = 3$ is precisely the threshold to make the $\delta^2$ term $\calO(1)$. Again we have used the fact $|z|\lesssim\sqrt{W}$.
		
		Second, we have \eq{ \Re{ \left\langle C,A\right\rangle _{\text{HS}}} &= 2F_{j-1}^{2}\Re{ \left\langle T_{j-1}\Gamma_{j-1}^{-1}T_{j-1}^{\ast},V_{j}\right\rangle _{\text{HS}}} } which is bounded in a similar manner as the $B$ term above after using a Schwarz inequality.
		
		We conclude that, even without restricting to the event $M_\phi$, just from the construction of $F$ in \cref{eq:F_j}, \eq{\norm{\ee^{W\calR}}_\infty\lesssim \ee^{\calO(1)}\lesssim 1\,. }
		
		\subsection{The product of the two Jacobians}
		Recall that we have abused the notation in the sense that $\eta^\pm$ was both the map on the single hopping matrix $T_j\mapsto T_j^\pm$ and also the symbol for the collective map $\Mat_W(\CC)^n\to\Mat_W(\CC)^n$ for all hopping matrices. Hence, for a single $j$, based on the definition of $\delta$ in \cref{eq:delta} and the map $\eta$ as in \cref{eq:abstract eta} it would suffice to show \eq{\left|\det\left(\calD \eta^+\right)\right| \left|\det\left(\calD \eta^-\right)\right| \gtrsim \ee^{-\delta^2 W^3}
		} which is equivalent (multiplying $n$ such terms and using that $\delta^2W^3\sim \frac{1}{n}$) to \eq{\norm{\prod_{\sigma\in\Set{\pm}}\frac{1}{\sqrt{J^\sigma}}}_{L^\infty}\lesssim\ee^{\calO(1)}\lesssim 1\,.}
	
		We now use the calculation for $\calD \eta$ from \cref{sec:Calculation of the eta Jacobian}. When multiplying the two Jacobians, the scalar exponential $\exp\left(\pm\delta F_j\right)$ cancels between the two and we have, using the notation $M_A := A\otimes Q_A^\ast$ for brevity, \eq{
			\left|\det\left(\calD \eta^+\right)\right| \left|\det\left(\calD \eta^-\right)\right| &= \left|\det\left(\Id+\delta M_A \right)\right|\left|\det\left(\Id-\delta M_A \right)\right| \\
			&= 
			\left|\det\left(\Id-\delta^2|M_A |^2\right)\right| \\
			&= \exp\left(\tr\left(\log\left(\Id-\delta^2|M_A |^2\right)\right)\right)
		} where in the second equality we have used the fact that $|\det X | = |\det X^\ast|$ and in the last equality we have used $\delta\norm{M_A } < 1$ which was established in the proof of \cref{lem:eta is injective} in \cref{sec:Calculation of the eta Jacobian}.
		
		By the spectral theorem, since $\log(1-\alpha)\geq-\alpha/(1-\alpha)$ for $\alpha\in(0,1)$, we have as a relation on self-adjoint operators, $$ \log\left(\Id-\delta^2|M_A |^2\right) \geq -\delta^2|M_A |^2(\Id-\delta^2|M_A |^2)^{-1}\geq -2\delta^2 |M_A |^2,$$ where we used that $\delta \norm{M_A }< 1/2$, and hence the same monotonicity holds when taking the trace. We thus get \eq{ \left|\det\left(\calD \eta^+\right)\right| \left|\det\left(\calD \eta^-\right)\right| \geq \exp\left(-\frac12\delta^2\tr_{\RR^{2W^2}}\left(|M_A|^2\right)\right)\,. }
		Now we have $\tr(|u\otimes v^\ast|^2) = \norm{u}^2\norm{v}^2$ and, as was already remarked above, the Euclidean norm in $\RR^{2W^2}$ is the Hilbert-Schmidt norm in $\Mat_W(\CC)$. Thus we are left with $$ \tr_{\RR^{2W^2}}\left(|M_A|^2\right) = \HSn{A}^2\HSn{Q_A}^2\lesssim W^2 $$ where the last step follows by \cref{lem:technical estimate on A and Q_A}.

		\section{The $\ve$ bound: proof of \cref{eq:epsilon term condition}}
		\label{sec:the epsilon term bound}
		
	The goal of this section is to prove that
	\begin{equation}
	    \label{eq:impprobou}
	    \mathbb{P}[M_\phi^c]\le e^{-\Xi n},
	\end{equation}
	for some fixed $\Xi>0$, which gives exactly \cref{eq:epsilon term condition}. Here we also recall that
	\[
	M_\phi:=\Set{(\Gamma,T)|F\ge \phi n},
	\]
	for some $\phi\in (0,1)$, which we will choose later in this section, and $F=\sum_j F_j$, with $F_j$ being defined as in \cref{eq:F_j}.
	
	As a first step we notice that $\widetilde{M}_\phi\subset M_\phi$, with $\widetilde{M}_\phi$ defined as
	\begin{equation}
	    \label{eq:deftildeM}
	    \widetilde{M}_\phi:= \Set{(\Gamma,T)|\HSn{\Gamma_j^{-1}},\HSn{\Gamma_{j+1}},\HSn{V_{j+1}}, \HSn{T_j}^2\le KW, \mathrm{for}\,\, \mathrm{at}\,\, \mathrm{least}\,\,\phi n\,\, \mathrm{indices} }
	\end{equation}
	
	Then, by the definition
		\[
			\Gamma_{j+1}=V_{j+1}+z+T_j\Gamma_j^{-1}T_j^*
			\]
		for any $j$, and a ``pigeon hole principle'' we readily see that
	\begin{equation}
	    \label{eq:inclusion}
	    M_1\cap M_2 \cap M_3 \cap M_4\subset \widetilde{M}_\phi,
	\end{equation}
	since the map $V\mapsto\Gamma$ is measure preserving, where we defined
	\begin{equation}
	    \begin{split}
	        M_1:&=\Set{\HSn{(V_j-A_j)^{-1}}\le KW/(3C^2) \,\,\mathrm{for}\,\, \mathrm{at}\,\, \mathrm{least}\,\,6\phi n\,\, \mathrm{indices}}, \\
	        M_2:&=\Set{\HSn{V_{j+1}}\le KW/3 \,\,\mathrm{for}\,\, \mathrm{at}\,\, \mathrm{least}\,\,6\phi n\,\, \mathrm{indices}}, \\
	        M_3:&=\Set{\HSn{T_j}^2\le KW \,\,\mathrm{for}\,\, \mathrm{at}\,\, \mathrm{least}\,\,6\phi n\,\, \mathrm{indices}}, \\
	        M_4:&=\Set{\norm{T_j}\le C \,\,\mathrm{for}\,\, \mathrm{at}\,\, \mathrm{least}\,\,6\phi n\,\, \mathrm{indices}},
	    \end{split}
	\end{equation}
	where we now choose any $\phi \in (0,1/6)$ as in \cref{eq:choice of phi}. Here $0< C< \sqrt{K}$ is a fixed constant, with $K$ which will be chosen shortly, and $A_j=A_j(V_{j-1},T_{j-1},\dots,V_1, T_1)$. Note that here we used that the change of variables $V_j\to \Gamma_j$ is measure preserving. Additionally, by $\mathbb{P}$ we denote the joint probability measure of all the $V_j$'s and $T_j$'s.
	
	Combining all this we readily see that
	\[
	\mathbb{P}[M_\phi^c]\le \sum_{i=1}^4 \mathbb{P}[M_i^c].
	\]
	
	Hence, to conclude the proof of \cref{eq:epsilon term condition} we now separately show that $\mathbb{P}[M_i^c]$ is exponentially small in $n$ for any $i=1,2,3,4$. To make the presentation shorter we only present the bound for $M_1$, all the other estimates being completely analogous after replacing \cref{eq:impbound} with
	\begin{equation}
				\label{eq:aeple}
				\mathbb{P}\left(\norm{T_j}_{HS}^2> KW\right)\le e^{-W^2},
			\end{equation}
			and a similar well-known bound for $\norm{T_j}$ and $\HSn{V_{j+1}}$.
			
			%In particular, this bound is much smaller than \cref{eq:impbound}, hence in the remainder of the proof we just consider the bound for the second line of \cref{eq:need}.

		By \cite[Theorem 1]{doi:10.1142/S0219199717500286} we have that
		\begin{equation}
			\label{eq:impbound}
			\mathbb{P}_{V_j}\left(\norm{(V_j-A)^{-1}}_{HS}> K W/(3C^2)\right)\le \widetilde{C}\frac{C^2}{K},
		\end{equation}
		for some fixed constant $\widetilde{C}>0$, and for any $K\ge 1$ uniformly in deterministic matrices $A$, and $W$. Here $\mathbb{P}_{V_j}(E):=\mathbb{E}_{V_j}[\mathbf{1}_E]$, for any event $E$, where $\mathbb{E}_{V_j}$ denotes the expectation with respect to the measure of a single $V_j$.

			Next, using \cref{eq:impbound}, we get
			\begin{equation}
				\begin{split}
					\label{eq:impbneedbast}
					\mathbb{P}[M_1^c]&=\mathbb{P}\left(\norm{(V_j-A_j)^{-1}}_{HS}> KW/(3C^2)\,\, \mathrm{for}\,\, \mathrm{at}\,\, \mathrm{least}\,\,(1-6\phi)n\,\, \mathrm{indices}\right)  \\
					&=\sum_{m=(1-6\phi)n}^n\mathbb{P}\left(\norm{(V_j-A_j)^{-1}}_{HS}> KW/(3C^2)\,\, \mathrm{for}\,\, \mathrm{exactly}\,\,m\,\, \mathrm{indices}\right) \\
					&= \sum_{m=(1-6\phi)n}^n\sum_{S\subset \Set{1,\dots,n}\atop |S|=m}\mathbb{P}\left(\norm{(V_j-A_j)^{-1}}_{HS}> KW/(3C^2)\,\, \mathrm{for}\,\, j\in S \right)\\
					&\lesssim \sum_{m=(1-6\phi)n}^n\sum_{S\subset \Set{1,\dots,n}}\left(\begin{matrix} n \\ m\end{matrix}\right) \frac{(\widetilde{C}C^2)^{|S|}}{K^{|S|}}\le \frac{(2C^2\widetilde{C})^n}{K^{(1-6\phi)n}},
				\end{split}
			\end{equation}
			where we used that we can perform the $V_j$-integration one by one for any realization of $A_j$ (starting from the largest index and proceeding in a decreasing order), and that \cref{eq:impbound} holds uniformly for any fixed $A_j$. This concludes the proof of \cref{eq:impprobou} choosing $K$ sufficiently large so that
			\[
			\frac{(2C^2\widetilde{C})^n}{K^{(1-6\phi)n}}\le e^{-\Xi n}.
			\]
			
			%	\corr UP TO HERE! \nc
			
			%combining \cref{eq:newbgood},\cref{eq:need}, and \cref{eq:impbneedbast}, and using that
			%\[
			%\mathbb{P}\left(\norm{V_j}_{HS}> KW, \mathrm{for}\,\, \mathrm{at}\,\, \mathrm{least}\,\,n-m_1\,\, \mathrm{indices}\right)\le \sum_{m=n-m_1}^n \left(\begin{matrix} n \\ m\end{matrix}\right)  \ee^{-\frac{1}{4}mK^2 W^2},
			%\]
			%	which follows by direct computations using that $V_j$ are i.i.d. GUE/GOE matrices,
			
			%Finally, combining \cref{eq:need} with \cref{eq:impbneedbast}, and that similar bounds holds for the last three lines of \cref{eq:need} using \cref{eq:aeple} instead of \cref{eq:impbound}, we conclude
		%	\begin{equation}
		%\begin{split}
		%			&\mathbb{P}\left(\left[\HSn{(V_j-A_j)^{-1}}, \HSn{V_{j+1}},\HSn{T_j}^2\le KW, \norm{T_j}\le C, \, \mathrm{for}\,\, \geq\,\, \,\,\phi\,\, \mathrm{indices}\right]^c\right) \\
		%			&\lesssim\sum_{m=n-\phi}^n \left(\begin{matrix} n \\ m\end{matrix}\right) \left(\frac{1}{K^m}+e^{-\frac{1}{4}m W^2}\right)\lesssim \frac{1}{K^{n-\phi}}\sum_{m=0}^n \left(\begin{matrix} n \\ m\end{matrix}\right)\le  \frac{2^n}{K^{n-\phi}}.
		%		\end{split}
		%	\end{equation}
		%	Here we used that $e^{W^2}\ge K$ for $W$ large enough. This concludes the proof of \cref{eq:epsilon term condition}.

		\section{Generalizations}
		\label{sec:Generalizations}
		
		In this section we discuss other and more general models to which our proof applies. We begin by outlining the most general class of models we can treat, and then point which special choices correspond to models of interest.
		
		\subsection{Mixture of Gaussian vectors taking values in general vector spaces}
		
		Let $\FF \in \Set{\RR,\CC}$ and for  $\ell\in\Set{1,2}$, let $\calV_\ell$ be two $\RR$-vector spaces of dimension $W^2\lesssim \dim \calV_\ell\lesssim W^2$. Assume further that there are $\RR$-linear injections $i_2:\calV_2\to\Mat_W(\FF)$ and $i_1:\calV_1\to\Herm_W(\FF)$.
		
		\begin{defn}[Mixture of Gaussian measures] Let $f:\RR^m\to(0,\infty)$ be a  density function (associated to a probability measure which is absolutely continuous w.r.t. to the Lebesgue measure on $\RR^m$). We say that $f$ is a mixture of Gaussians iff there is a positive measure $\mu$ on $(0,\infty)$ (which depends on $f$, \emph{but not on $m$}) such that \eql{\label{eq:mixture of Gaussians}
				f(v) = \int_{\lambda=0}^{\infty}\ee^{-\lambda \sqrt{m} \norm{v}^2}\dif{\mu(\lambda)}\qquad(v\in\RR^m)
			} where we use the Euclidean norm on $\RR^m$.  
		\end{defn}
		
		Let $p_\ell$ be a density on $\calV_\ell$ which is a mixture of Gaussians as in \cref{eq:mixture of Gaussians}, and further assume that the associated measures $\mu_\ell$ have support of diameter independent of $n,W$, i.e., \eql{\label{eq:assumption on support of generalized density}
		\supp(\mu_\ell) \subseteq [0,D_\ell]
		%	\ex{\HSn{A}^2} \equiv \frac{\int_{A\in i_\ell(\calV_\ell)}p_\ell(A) \HSn{A}^2\dif{A}}{\int_{A\in i_\ell(\calV_\ell)}p_\ell(A) \dif{A}} = D_\ell W^2 
		} 
		for some constant $D_\ell>0$ independent of $n,W$.% Here we mean by $\dif{A}$ the Lebesgue measure on $\calV_\ell$ pushed to $\Mat_W(\FF)$ or $\Herm_W(\FF)$ via $i_\ell$. 
		\begin{rem}
		    The same proof generalizes also to measures $\mu_\ell$ which do not have compact support, but sufficient decay at infinity so as to guarantee that the $L^\infty$ estimates further below on the remainder term $\calR$ go through, when replaced by integration. We do not pursue this further generalization here, but point out that \cref{eq:assumption on support of generalized density} heavily restricts the class of measures possible, in the sense that without lifting it, since we treat this coupling $\lambda$ as quenched, we are effectively taking the Ginibre, GUE distribution on each matrix but allowing the variances to vary with $j$. Since there are classes of interesting models covered if \cref{eq:assumption on support of generalized density} were lifted, we phrased the condition as a mixture of Gaussians. An example of such a model is: $$ M\mapsto \frac{1}{Z}\exp\left(-\xi\HSn{M}^\alpha\right) $$ with $\alpha\in(0,2)$.
		\end{rem}
		
		We define now the random Hamiltonian $H$ which is an $nW\times nW$ $\FF$-valued matrix as in \cref{eq:Hamiltonian}, but generalized so that $\Set{V_j}_{j=1}^n$ is an i.i.d. sequence of matrices each taking value in $i_1(\calV_1)\subseteq \Herm_W(\FF)$ and distributed according to $p_1$ and $\Set{T_j}_{j=1}^{n-1}$ is an i.i.d. sequence of matrices each taking values in $i_2(\calV_2)\subseteq \Mat_W(\FF)$ and distributed according to $p_2$. Since we required that each element in $i_1(\calV_1)$ is self-adjoint, $H$ itself is a self-adjoint matrix over $\FF$. Explicitly, the distribution of $H$ is given as follows. For any measurable $f:\Mat_\FF(W)\to\CC$, \eq{
			\ex{f(H)} \equiv \frac{\int_{V_1,\dots,V_n\in i_1(\calV_1);\, T_1,\dots,T_{n-1}\in i_2(\calV_2)}\left(\prod_{j=1}^{n-1} p_1(V_j)p_2(T_j)\right)p_1(V_n) f(H)\dif{V}\dif{T}}{\int_{V_1,\dots,V_n\in i_1(\calV_1);\, T_1,\dots,T_{n-1}\in i_2(\calV_2)}\left(\prod_{j=1}^{n-1} p_1(V_j)p_2(T_j)\right)p_1(V_n) \dif{V}\dif{T}} 
		} where by $\dif{V}\dif{T}$ we mean the Lebesgue measures $\dif{V_1}\dots\dif{V_n}\dif{T_1}\dots\dif{T_{n-1}}$.
		
		%There exists an $s_0\in(0,1)$ such that for all $s\in(0,s_0)$ and $M\in(0,\infty)$ with $M\ll \sqrt{W}$ if $z\in\RR$ with $|z|<M$, there are $C<\infty,\mu>0$ independent of $n,W$ with
	
		\begin{thm}[Generalization of \cref{thm:main localization theorem}]\label{thm:general localization theorem}
		Assume that $H$ is distributed as detailed in the present section. Assume further that the large deviation estimate \cref{eq:aeple} holds for both $p_1$ and $p_2$,\footnote{It is actually enough that \cref{eq:aeple} holds with $e^{-W^2}$ replaced by the inverse of a large constant $K$, independent of $W$, as in \cref{eq:impbound} and that \cref{eq:impbound} holds for $p_1$.} Then there exists an $s_0\in(0,1)$ such that for all $s\in(0,s_0)$ and $z\in\RR$ with $|z|<M$, where $M\in(0,\infty)$ and $M\lesssim \sqrt{W}$, there exist $C<\infty,\mu>0$, independent of $n$ and $W$, such that
		\eql{
			\label{eq:generalized main fractional moment bound}
			\ex{\norm{G(x,y;z)}^s} \leq W^C \exp\left(-\mu \frac{|x-y|}{W^\sharp}\right)\qquad(x,y\in\Set{1,\dots,n})
		}
		with $\sharp=3$.
		\end{thm}
		\begin{proof}[Sketch of proof]
			The first step is to expand all the factors of $p_1,p_2$ using \cref{eq:mixture of Gaussians}. Once this is done, conditioned on all factors the variances $\lambda_{\ell;j}$, where $\ell=1$ correspond to $V_j$ and $\ell=2$ corresponds to $T_j$, we get the following distribution on $H$. For any measurable $f:\Mat_W(\FF)\to\CC$, \eq{
				\tilde{\EE}\left[f(H)\right] \equiv \frac{\int_{V_1,\dots,V_n\in i_1(\calV_1);\, T_1,\dots,T_{n-1}\in i_2(\calV_2)}\ee^{-W\tr\left(\lambda_{1;n}|V_n|^2+\sum_{j=1}^{n-1}\lambda_{1;j}|V_j|^2+\lambda_{2;j}|T_j|^2\right)} f(H)\dif{V}\dif{T}}{\int_{V_1,\dots,V_n\in i_1(\calV_1);\, T_1,\dots,T_{n-1}\in i_2(\calV_2)}\ee^{-W\tr\left(\lambda_{1;n}|V_n|^2+\sum_{j=1}^{n-1}\lambda_{1;j}|V_j|^2+\lambda_{2;j}|T_j|^2\right)} \dif{V}\dif{T}} \,.
			}
			The proof as outlined above now goes through verbatim the same way, with minor modifications. Those modifications are as follows: instead of working with $\Mat_W(\CC)\cong\RR^{2W^2}$ for the change of variables for the $T_j$'s in \cref{sec:Calculation of the eta Jacobian}, we will have $\RR^{\dim\calV_2}$. Furthermore, the constants in the large deviation estimates in \cref{sec:the epsilon term bound} will also change. The assumption on the support of the measures $\mu_\ell$, \cref{eq:assumption on support of generalized density}, will guarantee that the estimates in \cref{sec:the beta term bound} go through (the $\calR$ term will contain $\lambda_{\ell;j}$ factors) via a simple $L^\infty$ bound. 
		\end{proof}
		
		\subsection{Models of interest}
		\begin{enumerate}
			\item To get the proper ``RBM'' model rather than the Wegner W-orbital model we have analyzed here, one replaces the Ginibre distribution used above for $T_j$'s, which takes values in the vector space $\Mat_W(\CC)$ with a Gaussian distribution on the vector space of triangular $W\times W$ complex matrices. The distribution of the $V_j$'s is unaltered.
			\item To get real-valued rather than complex-valued matrices (and hence replace the GUE distribution of the $V_j$'s with a GOE distribution, and use the real Ginibre distribution on $\Mat_W(\RR)$ for the $T_j$'s), use the vector spaces $\calV_1 := \RR^{W^2}, \calV_2 := \RR^{W^2}$.
%			\item \color{red}If \cref{eq:assumption on support of generalized density} were lifted as detailed in the remark above, one could deal with a model that has density proportional to $$ \exp\left(-W\HSn{M}\right)\,, $$ assuming one established \cref{eq:aeple,eq:impbound} hold for it as well. We leave this for future work.\nc
		\end{enumerate}

			\appendix
			
			\section{Technical results}
			\label{sec:technical}
			Here we present proofs for some of the technical lemmata.
			\begin{proof}[Proof of \cref{lem:log variance}]
				It is enough to establish
				\eq{ \frac{1}{s}\int_{0}^{s}f_{r,s}(q)\,\VV_{q}\left[\log\left(Y\right)\right]\dif{q} = \log\left(\frac{\EE\left[Y^s\right]^{r/s}}{\EE\left[Y^r\right]}\right)\,.
				}
				Starting from the left hand side, we use the identity \eq{
					\partial_q^2 \log\left(\ex{\exp\left(q X\right)}\right) = \VV_q[X]
				} 
				and place it into the integral. Separating the integration over $[0,s]$ to $[0,r]$ and $[r,s]$, the function $f_{r,s}$ simplifies and may then be integrated by parts (twice, for each segment of integration), yielding the right hand side.
			\end{proof}
			\begin{proof}[Proof of \cref{lem:MW lower bound on fluctuations}]
				We have the following chain of inequalities
				\eq{
					\mathbb{P}\left[\left[-\alpha,\alpha\right]^{c}\right] & =\mathbb{P}\left[\left(-\infty,-\alpha\right)\right]+\mathbb{P}\left[\left(\alpha,\infty\right)\right]\\
					& \geq\mathbb{P}\left[\left(-\infty,-a\right]\right]+\mathbb{P}\left[\left[a,\infty\right)\right] \ ,
					\intertext{since $a>\alpha$.  By the arithmetic mean-geometric mean inequality, this is further bounded below} 
					%& \geq\mathbb{P}\left[\left[-b,-a\right]\right]+\mathbb{P}\left[\left[a,b\right]\right]\\
					& \geq2\sqrt{\mathbb{P}\left[\left(-\infty,-a\right]\right]\mathbb{P}\left[\left[a,\infty\right)\right]}\\
					& \geq\frac{2}{\beta}\left(\mathbb{P}\left[\left[-\alpha,\alpha\right]\right]-\varepsilon\right)\\
					  &=\frac{2}{\beta}\left(1-\mathbb{P}\left[\left[-\alpha,\alpha\right]^{c}\right]-\varepsilon\right) \ ,
				}
				where in the middle step we have applied the hypothesis \cref{eq:MW hypothesis}.  The final expression implies
				\eq{
					\mathbb{P}\left[\left[-\alpha,\alpha\right]^{c}\right]  \geq  \frac{1-\varepsilon}{1+\frac{1}{2}\beta}\,.
				}
				We conclude with a simple Markov inequality
				\eq{
					\mathbb{E}\left[X^{2}\right] \geq \alpha^{2}\mathbb{P}\left[\Set{\left|X\right|\geq\alpha}\right]\geq\alpha^{2}\mathbb{P}\left[\Set{\left|X\right|>\alpha}\right]\,. \qquad  \qedhere
				}
			\end{proof}	
			
			\begingroup
			\let\itshape\upshape
			\printbibliography
			\endgroup
		\end{document}